\documentclass{IEEEtran}

\usepackage{setspace}
\usepackage{amsmath,amssymb}    
\usepackage[dvips]{graphicx} 
\usepackage{verbatim}   
\usepackage{color}      
\usepackage{subfigure}
\usepackage{epsfig}
\usepackage{float}
\usepackage{mathtools}
\usepackage{cite}

\newcommand{\bi}{\begin{itemize}}
\newcommand{\ei}{\end{itemize}}
\newcommand{\beq}{\begin{equation}}
\newcommand{\eeq}{\end{equation}}
\newcommand{\bqn}{\begin{eqnarray*}}
\newcommand{\eqn}{\end{eqnarray*}}
\newcommand{\ba}{\begin{array}}
\newcommand{\ea}{\end{array}}
\newcommand{\bs}{\begin{small}}
\newcommand{\es}{\end{small}}
\newcommand{\nn}{\nonumber}

\newcommand{\yd}{{\cal D}}

\newtheorem{theorem}{Theorem}[section]
\newtheorem{lemma}[theorem]{Lemma}
\newtheorem{proposition}[theorem]{Proposition}

\newtheorem{assumption}[theorem]{Assumption}

\newcommand{\qed}{\nobreak \ifvmode \relax \else
      \ifdim\lastskip<1.5em \hskip-\lastskip
      \hskip1.5em plus0em minus0.5em \fi \nobreak
      \vrule height0.75em width0.5em depth0.25em\fi}

\DeclareMathOperator*{\argmin}{arg\,min}
\DeclareMathOperator*{\argmax}{arg\,max}
\def\defeq{{\stackrel{\Delta}{=}}}

\newcommand{\bm}{\boldmath}
\newcommand{\uu}{\mbox{\bm $u$}}

\newcommand{\bpsi}{\mbox{\bm $\psi$}}
\newcommand{\vv}{\mbox{\bm $v$}}
\newcommand{\cc}{\mbox{\bm $c$}}
\newcommand{\DD}{\mbox{\bm $D$}}

\begin{document}

\title{Target Localization in Wireless Sensor Networks using Error Correcting Codes}
\author{Aditya Vempaty\thanks{This work was supported in part by CASE: The Center for Advanced Systems and Engineering, a NYSTAR center for advanced technology at Syracuse University; AFOSR under Grant FA9550-10-1-0263; the National Science of Council (NSC) of Taiwan under Grants NSC NSC 102-2221-E-011 -006 -MY3 and NSC 101-2221-E-011-069-MY3. Parts of this work were presented at the 38th IEEE
International Conference on Acoustics, Speech, and Signal Processing (ICASSP~2013) and at the 13th International Symposium on Communications and Information Technologies (ISCIT~2013).

Part of Y. S. Han's work was completed during his visit of Syracuse University from 2012 to 2013. 

A. Vempaty and P. K. Varshney are with the  Dept. of Electrical Engineering and Computer  Science, Syracuse University, Syracuse, NY USA (email:\{avempaty, varshney\}@syr.edu).

Y. S. Han is with the Dept. of Electrical Engineering, National Taiwan University of Science and Technology, Taipei, Taiwan (email: yshan@mail.ntust.edu.tw).

Copyright (c) 2013 IEEE. Personal use of this material is permitted.  However, permission to use this material for any other purposes must be obtained from the IEEE by sending a request to pubs-permissions@ieee.org.}, {\it Student Member}, {\it IEEE}, Yunghsiang S. Han, {\it Fellow}, {\it IEEE},  Pramod~K.~Varshney,~{\it Fellow}, {\it IEEE}} \pubid{}

\date{}
\maketitle

\begin{abstract}
In this work, we consider the task of target localization using quantized data in Wireless Sensor Networks (WSNs). We propose a computationally efficient localization scheme by modeling it as an iterative classification problem. We design coding theory based iterative approaches for target localization where at every iteration, the Fusion Center (FC) solves an $M$-ary hypothesis testing problem and decides the Region of Interest (ROI) for the next iteration. The coding theory based iterative approach works well even in the presence of Byzantine (malicious) sensors in the network.  We further consider the effect of non-ideal channels. We suggest the use of soft-decision decoding to compensate for the loss due to the presence of fading channels between the local sensors and the FC. We evaluate the performance of the proposed schemes in terms of the Byzantine fault tolerance capability and probability of detection of the target region. We also present performance bounds which help us in designing the system. We provide asymptotic analysis of the proposed schemes and show that the schemes achieve perfect region detection irrespective of the noise variance when the number of sensors tends to infinity. Our numerical results show that the proposed schemes provide a similar performance in terms of Mean Square Error (MSE) as compared to the traditional Maximum Likelihood Estimation (MLE) but are computationally much more efficient and are resilient to errors due to Byzantines and non-ideal channels. \end{abstract}

\begin{keywords}
Target Localization, Wireless Sensor Networks, Error Correcting Codes, Byzantines
\end{keywords}

\section{Introduction}
\label{intro}
Wireless sensor networks (WSNs) have been extensively employed to monitor a region of interest (ROI) for reliable detection/estimation/tracking of events \cite{akyildiz_commag02,niu_tsp06,Wang_TMC12,Veeravalli&Varshney_12}. In this work, we focus on target localization in WSNs. Localization techniques proposed in the literature for sensor networks include direction of arrival (DOA), time of arrival (TOA) and time-difference of arrival (TDOA) based methods \cite{molnar_icassp01}\cite{Chen&etal:01Icassp}. Recent research has focused on developing techniques which do not suffer from imperfect time synchronization. Received signal strength based methods, which do not suffer from imperfect synchronization and/or extensive processing, have been proposed which employ least-squares or maximum-likelihood (ML) based source localization techniques \cite{hu_eurasip03}\cite{hu_tsp05}.  In WSNs, due to power and bandwidth constraints, each sensor, instead of sending its raw data, sends quantized data to a central observer or Fusion Center (FC). The FC combines these local sensors' data to estimate the target location. 

Secure localization is very important as potential malicious sensors may attempt to disrupt the network and diminish its capability. Only in the recent past, researchers have investigated the problem of security threats \cite{HKD06} on sensor  networks. We focus on one particular class of security attacks, known as the Byzantine data attack \cite{Vempaty_SPM13} (also referred to as the Data Falsification Attack). A Byzantine attack involves malicious sensors within the network which send false information to the FC to disrupt the global inference process.  In our previous work \cite{Vempaty_tsp}, we have analyzed target localization in WSNs in the presence of Byzantines. By considering the Posterior Cram\'{e}r Rao bound or Posterior Fisher Information as the performance metric, we analyzed the degradation in system performance in the presence of Byzantines. We showed that the FC becomes `blind' to the local sensor's data when the fraction of Byzantines is greater than $50\%$. When the FC becomes `blind', it is not able to use any information received from the local sensors and estimates the target location based only on prior information. In order to make the network robust to such attacks, we considered mitigation techniques. We proposed a Byzantine identification scheme which observes the sensors' behavior over time and identifies the malicious sensors. We also proposed a dynamic non-identical threshold design for the network which makes the Byzantines `ineffective'. 

An important element of WSNs is the presence of non-ideal wireless channels between sensors and the FC~\cite{varshney_spmag06}\cite{Pottie&Kaiser:00ACM}. These non-ideal channels corrupt the quantized data sent by the local sensors to the FC. This causes errors which deteriorates the inference performance at the FC. One way to handle the channel errors is to use error correcting codes~\cite{islam}\cite{LIN83}. In \cite{Luo&Min:13IJDSN}, target localization based on maximum likelihood estimation at the FC was considered and coding techniques were proposed to handle the effect of imperfect channels between sensors and fusion center. 

In this work, we propose the use of coding theory techniques to estimate the location of the target in WSNs. In our preliminary work \cite{Vempaty_ICASSP13_loc}\cite{Vempaty_ISCIT13}, we have shown the feasibility of our approach by providing simulation/numerical results. In this paper, we develop the fundamental theory and derive asymptotic performance results. We first consider the code design problem in the absence of channel errors and Byzantine faults. The proposed scheme models the localization problem as an iterative classification problem. The scheme provides a coarse estimate in a computationally efficient manner as compared to the traditional ML based approach. We present performance analysis of the proposed scheme in terms of detection probability of the correct region. We show analytically that the scheme achieves perfect performance in the asymptotic regime. We address the issues of Byzantines and channel errors subsequently and modify our scheme to handle them. The error correction capability of the coding theory based approach provides Byzantine fault tolerance capability and the use of soft-decoding at the FC provides tolerance to the channel errors. In the remainder of the paper, we refer to this coding theory based localization approach as ``coding approach". The schemes proposed in this paper show the benefit of adopting coding theory based techniques for signal processing applications.

The remainder of the paper is organized as follows: In Section \ref{prel}, we describe the system model used and lay out the assumptions made in the paper. We also present a brief overview of Distributed Classification Fusion using Error Correcting Codes (DCFECC) \cite{Wang_jsac05} and Distributed Classification Fusion using Soft-decision Decoding (DCSD) \cite{Wang_twc06} approaches. We propose our basic coding scheme for target localization in Section \ref{local_iter}. The performance of the proposed scheme in terms of region detection probability is analyzed in this section. We extend this scheme to the exclusion method based coding scheme in Section \ref{sec:byz} to mitigate the effect of Byzantines in the network. We present some numerical results showing the benefit of the proposed schemes compared to the traditional maximum likelihood based scheme. We also present a discussion on system design based on the performance analysis carried out in this section. We consider the presence of non-ideal channels in Section \ref{soft_dec} and modify our decoding rule to make it robust to fading channels. We conclude our paper in Section \ref{conc} with some discussion on possible future work.

\section{Preliminaries}
\label{prel}
\subsection{System model}
Let $N$ sensors be randomly deployed (not necessarily in a regular grid) in a WSN as shown in Fig. \ref{model} to estimate the unknown location of a target at $\theta = [x_t, y_t]$, where $x_t$ and $y_t$ denote the coordinates of the target in a 2-D Cartesian plane. We assume that the location of the sensors is known to the Fusion Center (FC). We also assume that the signal radiated from this target follows an isotropic power attenuation model~\cite{niu_tsp06}. The signal amplitude $a_i$ received at the $i^{th}$ sensor is given by
\begin{equation}
\label{sys_mod}
a_i^2=P_0\left(\frac{d_0}{d_i}\right)^n,
\end{equation} 
where $P_0$ is the power measured at the reference distance $d_0$, $d_i \neq 0$ is the distance between the target and the $i^{th}$ sensor whose location is represented by $L_i=[x_i,y_i]$ for $i= 1,2 \cdots, N$ and $n$ is the path loss exponent. In this work, without loss of generality, we assume $d_0=1$ and $n=2$. The signal amplitude measured at each sensor is corrupted by independent and identically distributed (i.i.d.) zero-mean additive white noise with complementary cumulative distribution function given by $\bar{F}(\cdot;\sigma^2)$:
\begin{equation}
\label{eq:measurement}
s_i=a_i+n_i,
\end{equation}
where $s_i$ is the corrupted signal at the $i^{th}$ sensor and the noise $n_i \sim \bar{F}(\cdot;\sigma^2)$ with variance $\sigma^2$. 

\begin{figure}[htb]
\centering
\includegraphics[width = 4in,height=!]{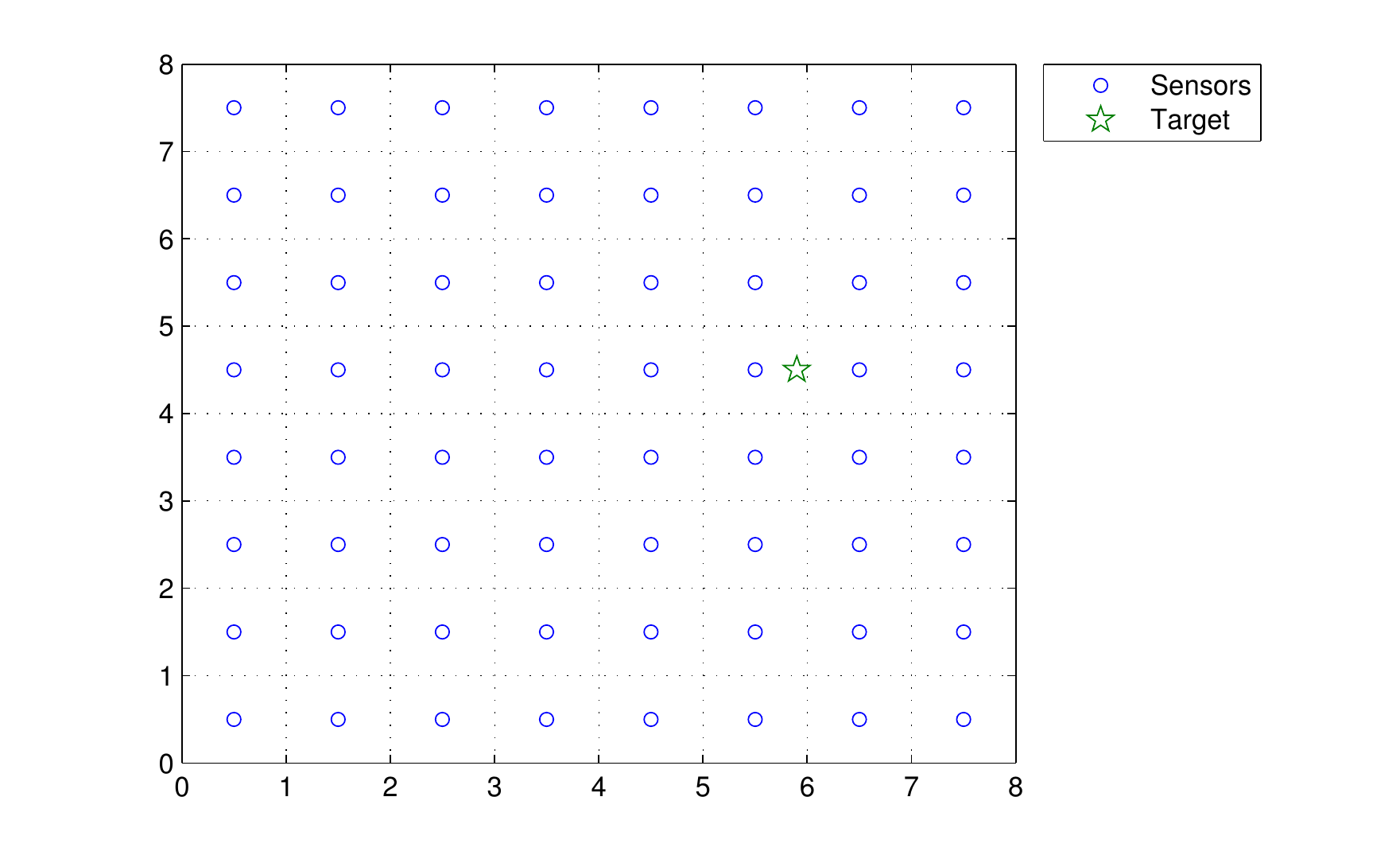}
\caption{Wireless sensor network layout for target localization}
\label{model}
\end{figure}

Due to energy and bandwidth constraints, the local sensors quantize their observations using threshold quantizers and send binary quantized data to the FC:
\begin{equation}
\label{threshold}
D_i =
\begin{cases}
0 & \text{$s_i < \eta_i$} \\
1 & \text{$s_i > \eta_i$}
\end{cases},
\end{equation}
where $D_i$ is the quantized data at the $i^{th}$ sensor and $\eta_i$ is the threshold used by the $i^{th}$ sensor for quantization. The FC fuses the data received from the local sensors and estimates the target location. Traditional target localization uses MLE~\cite{niu_tsp06}:
\begin{equation}
\label{MLE}
\hat{\theta}=\argmax_\theta{p(\uu|\theta)},
\end{equation}
where $\uu=[u_1, u_2,\cdots, u_N]$ is the vector of quantized observations received at the FC. As pointed out in the later sections, $\uu$ and $\DD$ can be different due to the presence of Byzantines and/or imperfect channels between local sensors and FC.

\subsection{An Overview of Distributed Classification Approaches}
\subsubsection{DCFECC \cite{Wang_jsac05}}
\label{DCFECC}
In this subsection, we give a brief overview of Distributed Classification Fusion using Error Correcting Codes (DCFECC) approach proposed in \cite{Wang_jsac05}. In \cite{Wang_jsac05}, the authors propose an approach for $M$-ary distributed classification using binary quantized data. After processing the observations locally, possibly in the presence of sensor faults, the $N$ local sensors transmit their local decisions to the FC. In the DCFECC approach, a code matrix $C$ is selected to perform both local decision and fault-tolerant fusion at the FC. The code matrix is an $M \times N$ matrix with elements $c_{(j+1)i} \in \{0,1\}$, $j=0,1, \cdots, M-1$ and $i=1,\cdots,N$. Each hypothesis $H_j$ is associated with a row in the code matrix $C$ and each column represents a binary decision rule at the local sensor. The optimal code matrix is designed off-line using techniques such as simulated annealing or cyclic column replacement \cite{Wang_jsac05}. After receiving the binary decisions $\uu$ from local sensors, the FC performs minimum Hamming distance based fusion and decides on the hypothesis $H_j$ for which the Hamming distance between row of $C$ corresponding to $H_j$ for $j=0,\cdots,M-1$ and the received vector $\uu$ is minimum.
It is important to note that the above scheme is under the assumption that $N>M$ and the performance of the scheme depends on the minimum Hamming distance $d_{min}$ of the code matrix $C$.

\subsubsection{DCSD\cite{Wang_twc06}}
\label{DCSD}
In this subsection, we present a brief overview of Distributed Classification using Soft-decision Decoding (DCSD) approach proposed in \cite{Wang_twc06}. This approach uses a soft-decision decoding rule as opposed to the hard-decision decoding rule used in DCFECC approach. The use of soft-decision decoding makes the system robust to fading channels between the sensors and the FC. The basic difference between the two approaches (DCFECC and DCSD) is the decoding rule. In DCFECC, the minimum Hamming distance rule is used. In the presence of fading channels, the received data at the FC is analog although the local sensors transmit quantized data based on the code matrix $C$ as described before. Then, the FC can use hard-decision decoding to determine the quantized data sent by the local sensors and use minimum Hamming distance rule to make a decision regarding the class. However, in \cite{Wang_twc06}, the authors show that the performance can deteriorate when hard-decision decoding is used. Instead, they propose a soft-decision decoding rule based on the channel statistics to make a decision regarding the class. We skip the derivation of the soft-decision decoding rule but present the decoding rule here for the case when binary quantizers are used at the local sensors, i.e., the elements of the code matrix are 0 or 1.

Let the analog data received at the FC from the local sensors be $\vv=[v_1,\cdots,v_N]$ when the local sensors transmit $\uu=[u_1,\cdots, u_N]$, where $u_i=0/1$ is decided by the code matrix $C$. For fading channels between the local sensors and the FC, $v_i$ and $u_i$ are related as follows

\begin{equation}
v_i=h_i(-1)^{u_i}\sqrt{E_b}+n_i,
\end{equation}
where $h_i$ is the channel gain that models the fading channel, $E_b$ is the energy per bit and $n_i$ is the zero mean additive white Gaussian noise. Define the reliability of the received data $v_i$ as

{\small \begin{eqnarray}
\psi_i=\ln{\frac{P(v_i|u_i=0)P(u_i=0|0)+P(v_i|u_i=1)P(u_i=1|0)}{P(v_i|u_i=0)P(u_i=0|1)+P(v_i|u_i=1)P(u_i=1|1)}}\label{reliability-DCFECC}
\end{eqnarray}}
for $i=\{1,\cdots,N\}$. Here $P(v_i|u_i)$ can be obtained from the statistical model of the fading channel considered and $P(u_i=d|s)$ for $s,d=\{0,1\}$ is the probability that the decision is $d$ given $s$ is present at the bit $i$ before local decision making and is given as follows

\begin{equation}
\label{eq:soft-dec1}
P(u_i=d|s)=\sum_{j=0}^{M-1}P(u_i=d|H_j)P_i(H_j|s).
\end{equation}
$P(u_i=d|H_j)$ depends on the code matrix while $P_i(H_j|s)$ is the probability that the hypothesis $H_j$ is true given $s$ is present at the bit $i$ (column $i$ of the code matrix) before local decision making, and can be expressed as
\begin{equation}
\label{eq:soft-dec2}
P_i(H_j|s)=\frac{P_i(s|H_j)}{\sum_{l=0}^{M-1}P_i(s|H_l)}
\end{equation}
where 
\begin{equation}
\label{eq:soft-dec3}
P_i(s|H_l)=
\begin{cases}
1,  \quad \text{if $c_{(l+1)i}=s$}\\
0,  \quad \text{if $c_{(l+1)i}\neq s$}
\end{cases}.
\end{equation}
Then the decoding rule is to decide the hypothesis $H_j$ where $j=\displaystyle\argmin_{0\leq j\leq M-1}d_F(\bpsi,\cc_{j+1})$. Here $d_F(\bpsi,\cc_{j+1})=\sum_{i=1}^N(\psi_i-(-1)^{c_{(j+1)i}})^2$ is the distance between $\bpsi=[\psi_1,\cdots,\psi_N]$ and $(j+1)^{th}$ row of $C$.

\section{Localization using iterative classification}
\label{local_iter}
In this section, we propose the localization scheme using iterative classification. Our algorithm is iterative in which at every iteration, the ROI is split into $M$ regions and an $M$-ary hypothesis test is performed at the FC to determine the ROI for the next iteration. The FC, through feedback, declares this region as the ROI for the next iteration. The $M$-ary hypothesis test solves a classification problem where each sensor sends binary quantized data based on a code matrix $C$. The code matrix is of size $M \times N$ with elements $c_{(j+1)i} \in \{0,1\}$, $j=0,1, \cdots, M-1$ and $i=1,\cdots,N$, where each row represents a possible region and each column $i$ represents $i^{th}$ sensor's binary decision rule. After receiving the binary decisions $\uu=[u_1,u_2,\cdots,u_N]$ from local sensors, the FC performs minimum Hamming distance based fusion. In this way, the search space for target location is reduced at every iteration and we stop the search based on a pre-determined stopping criterion. The optimal splitting of the ROI at every iteration depends on the topology of the network and the distribution of sensors in the network. For a given network topology, the optimal region split can be determined offline using k-means clustering \cite{Berkhin02surveyof} which yields Voronoi regions \cite{Aurenhammer:1991:VDS:116873.116880} containing equal number of sensors in every region. For instance, when the sensors are deployed in a regular grid, the optimal splitting is uniform as shown in Fig.~\ref{model_split}. In the remainder of the paper, we consider a symmetric sensor deployment such as a grid. Such a deployment results in a one-to-one correspondence between sensors across regions which is required in our derivations. Further discussion is provided in the later part of this section. In this section, the sensors are assumed to be benign and the channels between the local sensors and the FC are assumed to be ideal. Therefore, in this section, the binary decisions received at the FC are the same as the binary decisions made by the local sensors, i.e., $u_i=D_i$, for $i=1, \cdots, N$. We relax these assumptions in the later sections. The FC estimates the target location using the received data $\uu$. 

\begin{figure}[htb]
\centering
\includegraphics[width = 3.75in,height=!]{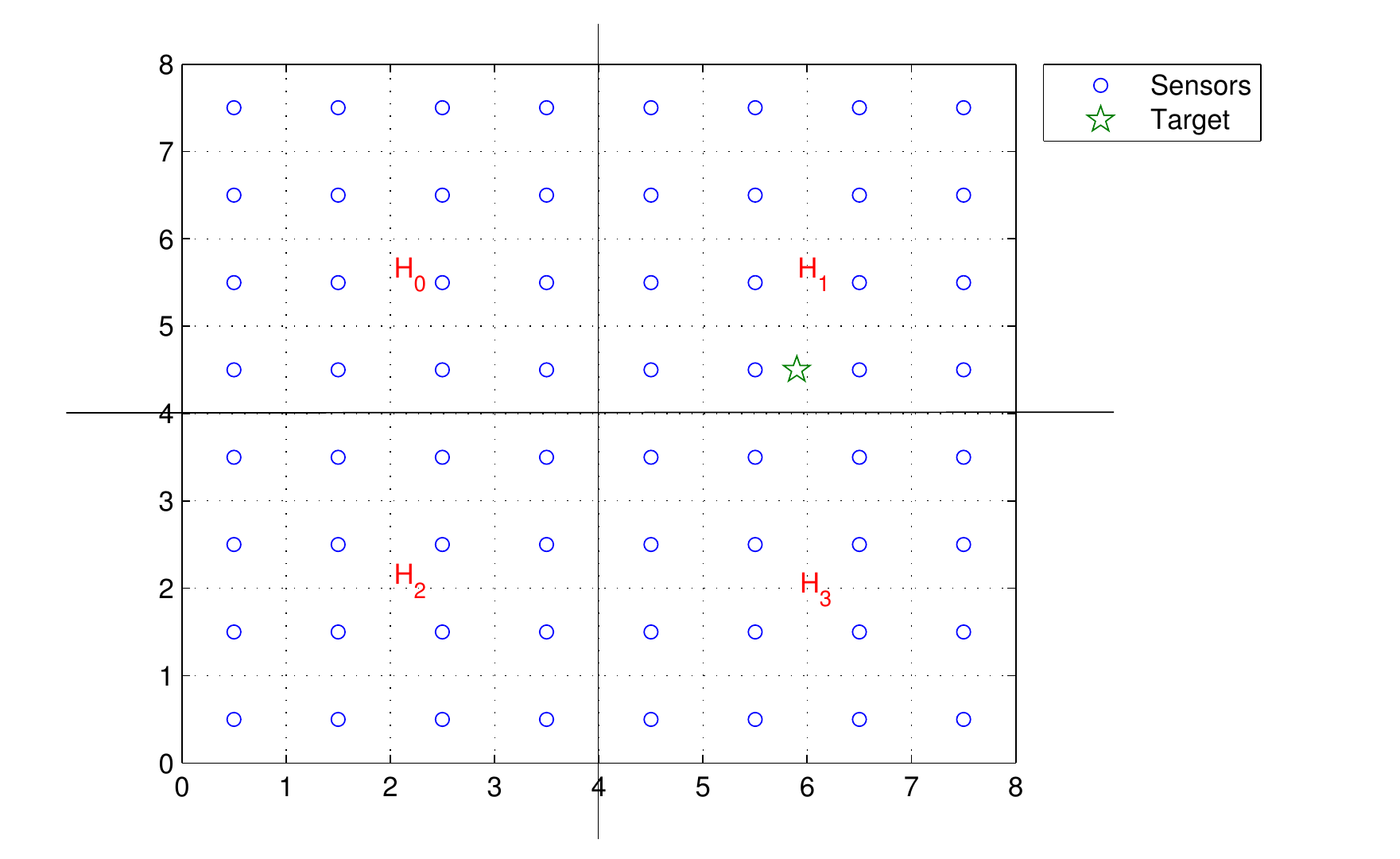}
\caption{Equal region splitting of the ROI for the $M$-hypothesis test}
\label{model_split}
\end{figure}

\subsection{Basic Coding Based Scheme}
\label{basic}
In this subsection, we present the basic coding based scheme for target localization. Since there are $N$ sensors which are split into $M$ regions, the number of sensors in the new ROI after every iteration is reduced by a factor of $M$. After $k$ iterations, the number of sensors in the ROI are $\frac{N}{M^k}$ and, therefore, the code matrix at the $(k+1)^{th}$ iteration would be of size $M \times \frac{N}{M^k}$.\footnote{We assume that $N$ is divisible by $M^k$ for $k=0,1,\ldots, \log_M N-1$.}  Since the code matrix should always have more columns than rows, $k^{stop} < \log_M{N}$, where $k^{stop}$ is the number of iterations after which the scheme terminates. After $k^{stop}$ iterations, there are only $\frac{N}{M^{k^{stop}}}$ sensors present in the ROI and a coarse estimate $\hat{\theta}=[\hat{\theta}_x, \hat{\theta}_y]$ of the target's location can be obtained by taking an average of locations of the $\frac{N}{M^{k^{stop}}}$ sensors present in the ROI:

\begin{eqnarray}
\hat{\theta}_x=\frac{M^{k^{stop}}}{N}\sum_{i\in ROI_{k^{stop}}}{x_i}\\ \text{and} \quad
\hat{\theta}_y=\frac{M^{k^{stop}}}{N}\sum_{i\in ROI_{k^{stop}}}{y_i},
\end{eqnarray}
where $ROI_{k^{stop}}$ is the ROI at the last step.

Since the scheme is iterative, the code matrix needs to be designed at every iteration. Observing the structure of our problem, we can design the code matrix in a simple and efficient way as described below. As pointed out before, the size of the code matrix $C^k$ at the $(k+1)^{th}$ iteration is $M \times \frac{N}{M^k}$, where $0\le k\le k^{stop}$. Each row of this code matrix $C^k$ represents a possible hypothesis described by a region in the ROI. Let $R_j^k$ denote the region represented by the hypothesis $H_j$ for $j=0,1,\cdots,M-1$ and let $S_j^k$ represent the set of sensors that lie in the region $R_j^k$. Also, for every sensor $i$, there is a unique corresponding region in which the sensor lies and the hypothesis of the region is represented as $r^k(i)$. It is easy to see that $S^k_j=\{i\in ROI_k|r^k(i)=j\}$. The code matrix is designed in such a way that for the $j^{th}$ row, only those sensors that are in $R_j^k$ have `1' as their elements in the code matrix. In other words, the elements of the code matrix are given by

\begin{equation}
\label{codematrix}
c^k_{(j+1)i}=
\begin{cases}
1 & \text{if $i \in \mathcal{S}^k_j$}\\
0 & \text{otherwise}
\end{cases},
\end{equation}

for $j=0,1,\cdots,M-1$ and $i \in ROI_k$.

The above construction can also be viewed as each sensor $i$ using a threshold $\eta^k_i$ for quantization (as described in \eqref{threshold}). Let each region $R_j^k$ correspond to a location $\theta^k_j$ for $j=0,1,\cdots,M-1$, which in our case is the center of the region $R_j^k$. Each sensor $i$ decides on a `1' if and only if the target lies in the region $R^k_{r^k(i)}$. Every sensor $i$, therefore, performs a binary hypothesis test described as follows:
\begin{eqnarray}
H_1:&	\text{$\theta^k \in R^k_{r^k(i)}$}\nonumber\\
H_0:&	\text{$\theta^k \notin R^k_{r^k(i)}$}.
\end{eqnarray}

If $d_{i,\theta^k_j}$ represents the Euclidean distance between the $i^{th}$ sensor and $\theta^k_j$ for $i=1,2,\cdots,N$ and $j=0,1,\cdots, M-1$, then $r^k(i)=\displaystyle\argmin_{l}{d_{i,\theta^k_l}}$. Therefore, the condition $\theta^k \in R^k_{r^k(i)}$ can be abstracted as a threshold $\eta^k_i$ on the local sensor signal amplitude given by
\begin{equation}
\label{eq:thresh}
\eta^k_i=\frac{\sqrt{P_0}}{d_{i,\theta^k_{r^k(i)}}}.
\end{equation}
This ensures that if the signal amplitude at the $i^{th}$ sensor is above the threshold $\eta^k_i$, then $\theta^k$ lies in region $R^k_{r^k(i)}$ leading to minimum distance decoding. 

\subsection{Performance Analysis}
\label{det}
In this subsection, we present the performance analysis of the proposed scheme. Although the performance metric in this framework is the Mean Square Error (MSE), it is difficult to obtain a closed form representation for MSE. Therefore, typically, one uses the bounds on MSE to characterize the performance of the estimator. In our previous works \cite{niu_tsp06,Vempaty_tsp}, we analytically derived the expressions of MSE bound (Posterior Cram\'{e}r Rao Lower Bound) on target localization under both non-adversarial \cite{niu_tsp06} and adversarial scenarios \cite{Vempaty_tsp}. An analytically tractable metric to analyze the performance of the proposed scheme is the probability of detection of the target region. It is an important metric when the final goal of the target localization task is to find the approximate region or neighborhood where the target lies rather than the true location itself. Since the final ROI could be one of the $M$ regions, a metric of interest is the probability of `zooming' into the correct region. In other words, it is the probability that the true location and the estimated location lie in the same region.

The final region of the estimated target location is the same as the true target location, if and only if we `zoom' into the correct region at every iteration of the proposed scheme. If $P_d^k$ denotes the detection probability (probability of correct classification) at the $(k+1)^{th}$ iteration, the overall detection probability is given by

\begin{align}
P_D=\prod_{k=0}^{k^{stop}}{P_d^k}\label{final_pd}.
\end{align}

\subsubsection*{Exact Analysis}
Let us consider the $(k+1)^{th}$ iteration and define the received vector at the FC as $\uu^k=[u^k_1,u^k_2,\cdots,u^k_{N_k}]$, where $N_k$ are the number of local sensors reporting their data to FC at $(k+1)^{th}$ iteration. Let $\yd^k_j$ be the decision region of $j^{th}$ hypothesis defined as follows:

$$\yd^k_j=\{ \uu^k| d_H(\uu^k,\cc^k_{j+1})\le d_H(\uu^k,\cc^k_{l+1}) \mbox{ for } 0\le l\le M-1\},$$
where $d_H(\cdot,\cdot)$ is the Hamming distance between two vectors, and $\cc^k_{j+1}$ is the codeword corresponding to hypothesis $j$ in code matrix $C^k$. Then define the reward  $r^{j,k}_{\uu^k}$ associated with the hypothesis $j$ as
\begin{eqnarray}
r^{j,k}_{\uu^k}=\left\{\begin{array}{ll}\frac{1}{q_{\uu^k}}& \mbox{ when } {\uu^k}\in \yd^k_j\\
0&\mbox{otherwise}
\end{array},\right.\label{reward}
\end{eqnarray}
where $q_{\uu^k}$ is the number of decision regions to whom $\uu^k$ belongs to. Note that $q_{\uu^k}$ can be greater than one when there is a tie at the FC. Under such scenarios when $q_{\uu^k} > 1$, we break the tie using random decision. Since the tie-breaking rule is to choose one of them randomly, which is successful with probability $\frac{1}{q_{\uu^k}}$, the reward is given by \eqref{reward}. According to \eqref{reward}, the detection probability at the $(k+1)^{th}$ iteration is given by
\begin{eqnarray}
\label{pdk}
P_d^k&=&\sum_{j=0}^{M-1}P(H_j^k)\sum_{\uu^k\in \{0,1\}^{N_k}}P(\uu^k|H_j^k)r_{\uu^k}^{j,k}\nonumber\\
&=&\frac{1}{M}\sum_{j=0}^{M-1}\sum_{\uu^k\in \yd_j^k}\left(\prod_{i=1}^{N_k}P(u^k_i|H_j^k)\right)\frac{1}{q_{\uu^k}},
\end{eqnarray}
where $P(u^k_i|H_j^k)$ denotes the probability that the sensor $i$ sends the bit $u^k_i \in \{0,1\}$, $i=1,2,\cdots,N_k$, when the true target is in the region $R_j^k$ corresponding to $H_j^k$ at the $(k+1)^{th}$ iteration. 

From the system model described before, we get

\begin{eqnarray}
\label{eq:localprob}
&P(u^k_i=1|H_j^k)=E_{\theta|H_j^k}\left[P(u^k_i=1|\theta,H_j^k)\right].
\end{eqnarray}
Since \eqref{eq:localprob} is complicated, it can be approximated using $\theta_j^k$ which is the center of the region $R_j^k$. \eqref{eq:localprob} now simplifies to 
\begin{equation}
P(u_i^k=1|H_j^k)\approx \bar{F}\left(\eta_i^k-a_{ij}^k;\sigma^2\right),
\end{equation}
where $\eta_i^k$ is the threshold used by the $i^{th}$ sensor at $k^{th}$ iteration, $\sigma^2$ is the noise variance, $a_{ij}^k$ is the signal amplitude received at the $i^{th}$ sensor when the target is at $\theta_j^k$ and $\bar{F}(x;\sigma^2)$ is the complementary cumulative distribution function of noise at the local sensors.

Using \eqref{final_pd}, the probability of detection of the target region can be found as the product of detection probabilities at every iteration $k$. It is clear from the derived expressions that the exact analysis of the detection probability is complicated and, therefore, we derive some analytical bounds on the performance of the proposed scheme.

\subsubsection*{Performance bounds}
In this section, we present the performance bounds on our proposed coding based localization scheme. For our analysis, we will use the lemmas in \cite{Yao_TIT'07}, which are stated here for the sake of completeness.

\begin{lemma}[\cite{Yao_TIT'07}] 
\label{lemma:bound1}
Let $\{Z_j\}_{j=1}^\infty$ be independent antipodal random variables with $Pr[Z_j=1]=q_j$ and $Pr[Z_j=-1]=1-q_j$. 
If $\lambda_m \defeq E[Z_1+\cdots+Z_m]/m < 0$, then

\begin{equation}
Pr\{Z_1+\cdots+Z_m \geq 0\}\leq(1-\lambda_m^2)^{m/2}.
\label{lambda_m}
\end{equation}
\end{lemma}

Using this lemma, we now present the performance bounds on our proposed scheme. 

\begin{lemma} 
\label{lemma:bound}
Let $\theta\in R_j^k$ be the fixed target location. Let $P_e^k(\theta)$ be the misclassification probability of the target region given $\theta$  at the $(k+1)^{th}$ iteration. For the received vector of $N_k=N/M^k$ observations at the $(k+1)^{th}$ iteration, $\uu^k=[u_1^k,\cdots,u_{N_k}^k]$, assume that for every $0 \leq j,l \leq M-1$ and $l\neq j$, 
\begin{equation}
\sum_{i\in S_j^k\cup S_l^k}{q_{i,j}^k} < \frac{N_k}{M}=\frac{N}{M^{k+1}},\label{condition}
\end{equation}
where $q_{i,j}^k=P\{z_{i,j}^k=1|\theta\}$, $z_{i,j}^k=2(u_i^k \oplus c_{(j+1)i}^k)-1$, and $C^k=\{c_{(j+1)i}^k\}$ is the code matrix used at the $(k+1)^{th}$ iteration. Then
\begin{eqnarray}
P_e^k(\theta)  &\le& \sum_{0 \leq l \leq M-1, l\neq j}\left(1-\frac{\left(\sum_{i\in S_j^k\cup S_l^k}(2q_{i,j}^k-1)\right)^2}{d_{m,k}^2}\right)^{d_{m,k}/2}\label{eq:bound1}\\
&\le&(M-1)\left(1-\left(\lambda_{j,\text{max}}^k(\theta)\right)^2\right)^{d_{m,k}/2}\label{eq:bound2},
\end{eqnarray}
where $d_{m,k}$ is the minimum Hamming distance of the code matrix $C^k$ given by $d_{m,k}=\frac{2N}{M^{k+1}}$ due to the structure of our code matrix and

\begin{equation}
\lambda_{j,\text{max}}^k(\theta)\defeq\max_{0\leq l\leq M-1,l\neq j}\frac{1}{d_{m,k}}\sum_{i\in S_j^k\cup S_l^k}(2q_{i,j}^k-1).
\end{equation}

\end{lemma}

\begin{proof}
Let $d_H(\cdot,\cdot)$ be the Hamming distance between two vectors, for fixed $\theta \in R_j^k$,
\begin{eqnarray}
&&P_e^k(\theta)\nn\\
&=&P\left\{\text{detected region} \neq R_j^k|\theta\right\}\nn\\
 &\leq& P\left\{d_H(\uu^k,\cc_{j+1}^k) \geq \min_{0\leq l \leq M-1,l\neq j} d_H(\uu^k,\cc_{l+1}^k)|\theta\right\}\nonumber\\
&\leq& \sum_{0\leq l\leq M-1, l\neq j}P\left\{d_H(\uu^k,\cc_{j+1}^k) \geq d_H(\uu^k,\cc_{l+1}^k)|\theta\right\}\nonumber\\
&=& \sum_{0\leq l\leq M-1, l\neq j} P\left\{\sum_{\{i\in[1,\cdots,N_k]:c_{(l+1)i}\neq c_{(j+1)i}\}}z_{i,j}^k \geq 0|\theta\right\}. \nn\\\label{p_d-bound}
\end{eqnarray}

Using the fact that $c_{(l+1)i}^k \neq c_{(j+1)i}^k$ for all $i\in S_j^k\cup S_l^k$, $l \neq j$, we can simplify the above equation. Also, observe that $\{z_{i,j}\}_{i=1}^{N_k}$ are independent across the sensors given $\theta$. According to (2) in~\cite{Yao_TIT'07}, 
\begin{eqnarray}
\lambda_m&=&\frac{1}{d_{m,k}}\sum_{i=1}^{N_k}(c_{(l+1)i}^k\oplus c_{(j+1)i}^k)(2q_{i,j}^k-1)\nn\\
&=&\frac{1}{d_{m,k}}\sum_{i\in S_j^k\cup S_l^k}(2q_{i,j}^k-1)\nn\\
&=&\frac{1}{d_{m,k}}\left(\sum_{i\in S_j^k\cup S_l^k}2q_{i,j}^k-\frac{2N_k}{M}\right)\label{lambda_m-2}
\end{eqnarray}
since $c_{(l+1)i}^k \neq c_{(j+1)i}^k$ for all $i\in S_j^k\cup S_l^k$, $l \neq j$. Here, we have used the fact that cardinality of $S_j^k = N_k/M$ for all $j$, and $S_j^k$ and $S_l^k$ are disjoint sets for all $l \neq j$.
Condition $\lambda_m<0$ of Lemma \ref{lemma:bound1} is then equivalent to condition~\eqref{condition}. Therefore, using Lemma \ref{lemma:bound1} and \eqref{lambda_m-2}, we have

\begin{eqnarray}
&&P\left\{\sum_{\{i\in[1,\cdots,N_k]:c_{(l+1)i}\neq c_{(j+1)i}\}} z_{i,j}^k \geq 0 |\theta \right\} \nn\\
& \leq & \left(1-\frac{\left(\sum_{i\in S_j^k\cup S_l^k}(2q_{i,j}^k-1)\right)^2}{d_{m,k}^2}\right)^{d_{m,k}/2}.\label{inf_bound}
\end{eqnarray}
Substituting \eqref{inf_bound} into \eqref{p_d-bound}, we have
\eqref{eq:bound1}. Note that condition~\eqref{condition} ($\lambda_m<0$) implies $\lambda_{j,\text{max}}^k(\theta)<0$ by definition. Hence, \eqref{eq:bound2} is   a direct consequence from \eqref{eq:bound1}.
\end{proof}

The probabilities $q_{i,j}^k=P\{u_i^k\neq c_{(j+1)i}^k|\theta\}$ can be easily computed as below. For $0 \leq j \leq M-1$ and $1 \leq i \leq N_k$, if $i \in S_j^k$,

\begin{eqnarray}
q_{i,j}^k&=&P\{u_i^k=0|\theta\}\nn\\
&=&1-\bar{F}\left(\eta_i^k-a_i;\sigma^2\right),\label{eq:prob1}
\end{eqnarray}
where $\eta_i^k$ is the threshold used by the $i^{th}$ sensor at $(k+1)^{th}$ iteration, $\sigma^2$ is the noise variance, $a_i$ is the amplitude received at the $i^{th}$ sensor given by \eqref{sys_mod} when the target is at $\theta$.
If $i \notin S_j^k$, $q_{i,j}^k=1-P\{u_i^k=0|\theta\}$.

Before we present our main theorem, for ease of analysis, we give an assumption that will be used in the theorem. Note that, our proposed scheme can still be applied to those WSNs where the assumption does not hold. 

\begin{assumption}
\label{assumption}
For any target location $\theta\in R_j^k$ and any $0\le k\le k^{stop}$, there exists a bijection function $f$ from $S_j^k$ to $S_l^k$, where $0\le l\le M-1$ and $l\neq j$, such that 
$$f(i_j)=i_l,$$
$$\eta_{i_j}^k=\eta_{i_l}^k,$$
and
$$d_{i_j}<d_{i_l},$$
where $i_j\in S_j^k$, $i_l\in S_l^k$, and $d_{i_j}$ ($d_{i_l}$) is the distance between $\theta$ and sensor $i_j$ ($i_l$).
\end{assumption}

One example of WSNs that satisfies this assumption is given in Fig.~\ref{symmetric_sensor}. For every sensor $i_j \in S_j^k$, due to symmetric region splitting, there exists a corresponding sensor $i_l \in S_l^k$ which is symmetrically located as described in the following: Join the centers of the two regions and draw a perpendicular bisector to this line as shown in Fig. \ref{symmetric_sensor}. The sensor $i_l \in S_l^k$ is the sensor located symmetrically to sensor $i_j$ on the other side of the line $L$. These are the sensors for which the thresholds are the same. In other words, due to the symmetric placement of the sensors, $\eta_{i_j}^k=\eta_{i_l}^k$ (c.f. \eqref{eq:thresh}). Clearly, when $\theta\in R_j^k$, $d_{i_j}<d_{i_l}$.

\begin{figure}[thb]
\centering
\includegraphics[width = 4in,height=!]{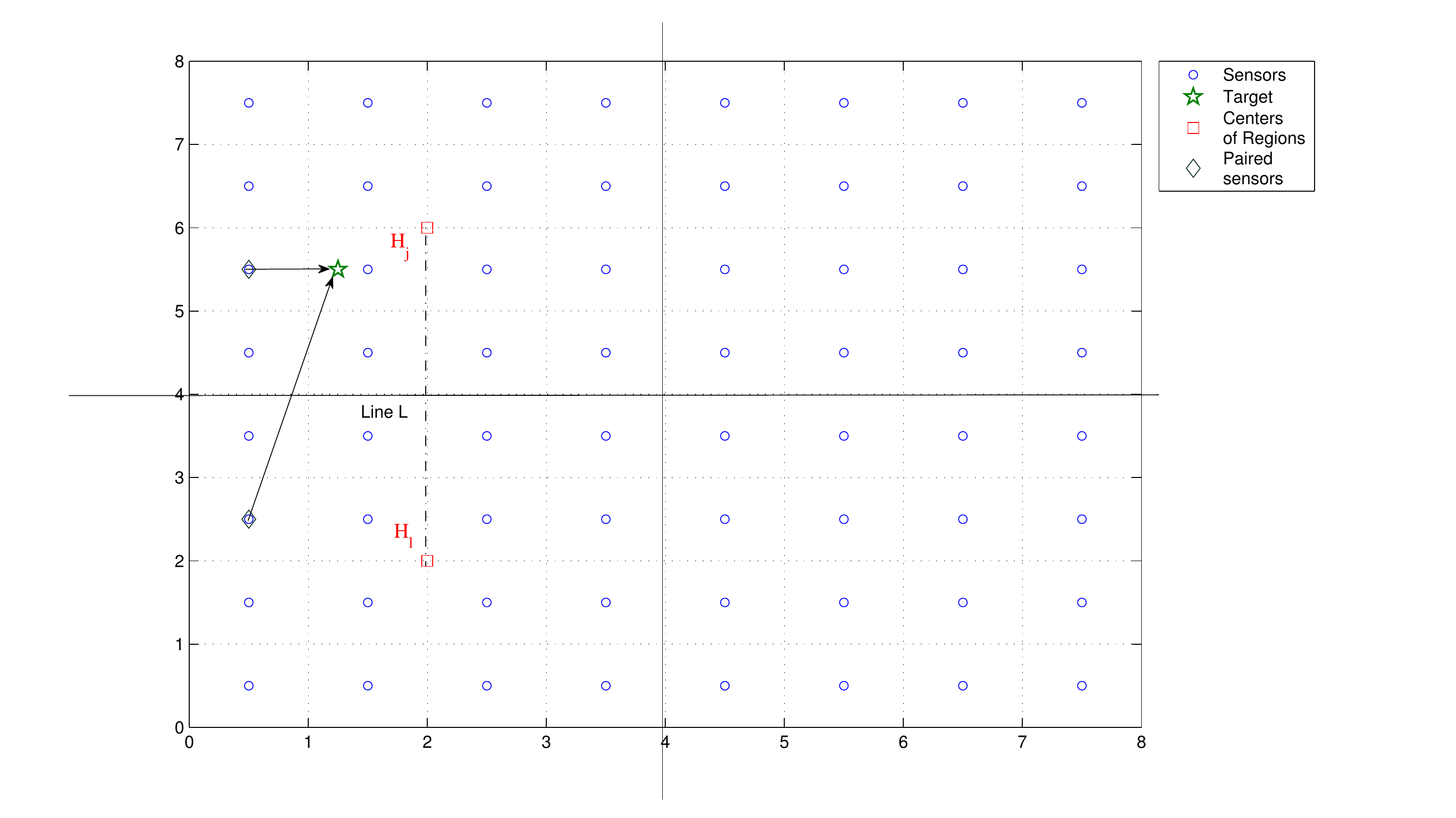}
\caption{ROI with an example set of paired sensors}
\label{symmetric_sensor}
\end{figure}

\begin{theorem}
\label{main-theorem}
Let $P_D$ be the probability of detection of the target region given by \eqref{final_pd}, where $P_d^k$ is the detection probability at the $(k+1)^{th}$ iteration. Under Assumption~\ref{assumption},
\begin{equation}
P_d^k \geq 1-(M-1)\left(1-(\lambda^k_{\text{max}})^2\right)^{d_{m,k}/2},
\label{main-equation}
\end{equation}
where 
$$\lambda^k_{\text{max}}\defeq\max_{0\le j\le M-1}\lambda^k_{j,\text{max}}$$
and 
$$\lambda^k_{j,\text{max}}\defeq\max_{\theta\in R_j^k}\lambda^k_{j,\text{max}}(\theta).$$
\end{theorem}
\begin{proof}
First we prove that condition~\eqref{condition} is satisfied by the proposed scheme for all $\theta$ when the noise variance, $\sigma^2<\infty$. Hence, the inequality~\eqref{eq:bound2} can be applied to the proposed scheme.
 The probabilities $q_{i,j}^k$ given by \eqref{eq:prob1} are
\begin{equation}
q_{i,j}^k=
\begin{cases}
1-\bar{F}\left(\eta_i^k-a_i;\sigma^2\right), &\text{for $i \in S_j^k$}\\
\bar{F}\left(\eta_i^k-a_i;\sigma^2\right), & \text{for $i \in S_l^k$}
\end{cases}.\label{eq:q-bound}
\end{equation}
By Assumption~\ref{assumption}, there exists a bijection function $f$ from $S_j^k$ to $S_l^k$. The sum $\sum_{i\in S_j^k\cup S_l^k}{q_{i,j}^k}$ of \eqref{condition} can be evaluated by considering pairwise summations as follows. Let us consider one such pair $(i_j\in S_j^k,f(i_j)=i_l\in S_l^k)$. Hence, their thresholds are $\eta_{i_j}^k=\eta_{i_l}^k=\eta$. Then, from \eqref{eq:q-bound},
\begin{eqnarray}
q_{i_j,j}^k+q_{i_l,j}^k&=&1-\bar{F}\left(\eta-a_{i_j};\sigma^2\right)+\bar{F}\left(\eta-a_{i_l};\sigma^2\right)\\
&=&1-\left[\bar{F}\left(\eta-a_{i_j};\sigma^2\right)-\bar{F}\left(\eta-a_{i_l};\sigma^2\right)\right].\nn\\\label{eq:pairwise_sum}
\end{eqnarray}
Now observe that, by the assumption, 
$$a_{i_j}=\frac{\sqrt{P_0}}{d_{i_j}}>\frac{\sqrt{P_0}}{d_{i_l}}=a_{i_l}$$ 
and, therefore, $\bar{F}\left(\eta-a_{i_j};\sigma^2\right)>\bar{F}\left(\eta-a_{i_l};\sigma^2\right)$ for all finite values of noise variance $\sigma^2$. From \eqref{eq:pairwise_sum}, the sum $q_{i_j,j}^k+q_{i_l,j}^k$ is strictly less than 1. Therefore, the sum $\sum_{i\in S_j^k\cup S_l^k}{q_{i,j}^k}<\frac{N_k}{M}=\frac{N}{M^{k+1}}=\frac{d_{m,k}}{2}$. Therefore, the condition in \eqref{condition} is satisfied for the code matrix used in this scheme. Hence, $P_e^k(\theta)$ can always be bounded by \eqref{eq:bound2}.

By using \eqref{eq:bound2}, $P_d^k$ can be bounded as follows:
\begin{eqnarray}
&&P_d^k\nn\\
&=&1-\sum_{j=0}^{M-1}P\{\theta\in R_j^k\}P\left\{\text{detected region} \neq R_j^k|\theta\in R_j^k\right\}\nn\\
&=&1-\frac{1}{M}\sum_{j=0}^{M-1}\int_{\theta}\nn\\
&&P\{\theta|\theta\in R_j^k\}P\left\{\text{detected region} \neq R_j^k|\theta,\theta\in R_j^k\right\}\ d\theta\nn\\
&=&1-\frac{1}{M}\sum_{j=0}^{M-1}\int_{\theta\in R_j^k}P\{\theta|\theta\in R_j^k\}P_e^k(\theta)\ d\theta\nn\\
&\ge& 1-\frac{1}{M}\sum_{j=0}^{M-1}\int_{\theta\in R_j^k}\nn\\
&&P\{\theta|\theta\in R_j^k\}(M-1)\left(1-\left(\lambda_{j,\text{max}}^k(\theta)\right)^2\right)^{d_{m,k}/2}\ d\theta\nn\\
&\ge& 1-\frac{M-1}{M}\sum_{j=0}^{M-1}\left(1-\left(\lambda_{j,\text{max}}^k\right)^2\right)^{d_{m,k}/2}\nn\\
&&\int_{\theta\in R_j^k}P\{\theta|\theta\in R_j^k\}\ d\theta\label{theorem:main-bound-2}\\
&\ge&1-\frac{M-1}{M}\sum_{j=0}^{M-1}\left(1-\left(\lambda_{\text{max}}^k\right)^2\right)^{d_{m,k}/2}\label{theorem:main-bound-1}\\
&=&1-(M-1)\left(1-\left(\lambda_{\text{max}}^k\right)^2\right)^{d_{m,k}/2}.\label{theorem:main-bound}
\end{eqnarray}
Both \eqref{theorem:main-bound-2} and \eqref{theorem:main-bound-1} are true since $\lambda_{j,\text{max}}^k<0$ and $\lambda_{\text{max}}^k<0$.
\end{proof}

Next we analyze the asymptotic performance of the scheme, i.e., we examine $P_D$ when $N$ approaches infinity.
\begin{theorem}
Under Assumption~\eqref{assumption}, $\displaystyle\lim_{N\rightarrow\infty} P_D=1$.
\end{theorem}
\begin{proof}
We have \begin{eqnarray}
\lambda^k_{j,\text{max}}&=&\max_{0\leq l\leq M-1,l\neq j}\frac{1}{d_{m,k}}\sum_{i\in S_j^k\cup S_l^k}(2q_{i,j}^k-1)\nn\\
&>&\frac{M^{k+1}}{2N}\sum_{i\in S_j^k\cup S_l^k}(-1)=-1
\end{eqnarray}
for all $0\le j\le M-1$ since not all $q^k_{i,j}=0$. Hence, by definition,   $\lambda^k_{\text{max}}$ is also greater than $-1$. Since $-1<\lambda^k_{\text{max}}<0$, we have $0<1-(\lambda^k_{\text{max}})^2<1$. 
Under the assumption that the number of iterations are finite, for a fixed number of regions $M$, we can analyze the performance of the proposed scheme under asymptotic regime.  Under this assumption, $d_{m,k}=\frac{2N}{M^{k+1}}$ grows linearly with the number of sensors $N$ for $0\le k\le k^{stop}$. Then
\begin{eqnarray}
\lim_{N\rightarrow\infty} P_D&=&\lim_{N\rightarrow\infty} \prod_{k=0}^{k^{stop}}P^k_d\nn\\
&\ge&\prod_{k=0}^{k^{stop}}\lim_{N\rightarrow\infty}\left[1-(M-1)(1-(\lambda^k_{\text{max}})^2)^{d_{m,k}/2}\right]\nn\\
&=&\prod_{k=0}^{k^{stop}}\left(1-(M-1)\lim_{N\rightarrow\infty}\left[(1-(\lambda^k_{\text{max}})^2)^{d_{m,k}/2}\right]\right)\nn\\
&=&\prod_{k=0}^{k^{stop}}[1-(M-1)0]\nn\\
&=&\prod_{k=0}^{k^{stop}}1=1.\nn
\end{eqnarray}
Hence,  the overall detection probability becomes `1' as the number of sensors $N$ goes to infinity. This shows that the proposed scheme asymptotically attains perfect region detection probability irrespective of the value of finite noise variance.
\end{proof}
Note that the above result also holds when $M$ increases with $N$ as long as $d_{m,k}=\frac{2N}{M^{k+1}}$ grows with the number of sensors $N$ for $0\leq k\leq k^{stop}$. In other words, our theory can be extended to scenarios when $M$ increases with $N$ as long as $\frac{N}{M^{k+1}} \to \infty$ as $N \to \infty$ for $0\leq k\leq k^{stop}$.

\subsection{Numerical Results}
We now present some numerical results which justify the analytical results presented in the previous subsection and provide some insights. In the previous subsection, we have observed that the performance of the basic coding scheme quantified by the probability of region detection asymptotically approaches `1' irrespective of the finite noise variance. Fig.~\ref{Pd_sigma} shows that the region detection probability approaches `1' uniformly as the number of sensors approaches infinity for Gaussian sensor observation noise with variance $\sigma^2$. Observe that for a fixed noise variance, the region detection probability increases with increase in the number of sensors. This can also be observed from Table \ref{table:pd}. Also, for a fixed number of sensors, the region detection probability decreases with $\sigma$ when the number of sensors is small. But when the number of sensors is large, the reduction in region detection probability with $\sigma$ is negligible and as $N \to \infty$, the region detection probability converges to 1.

\begin{figure}[htb]
\centering
\includegraphics[width = 3.5in,height=!]{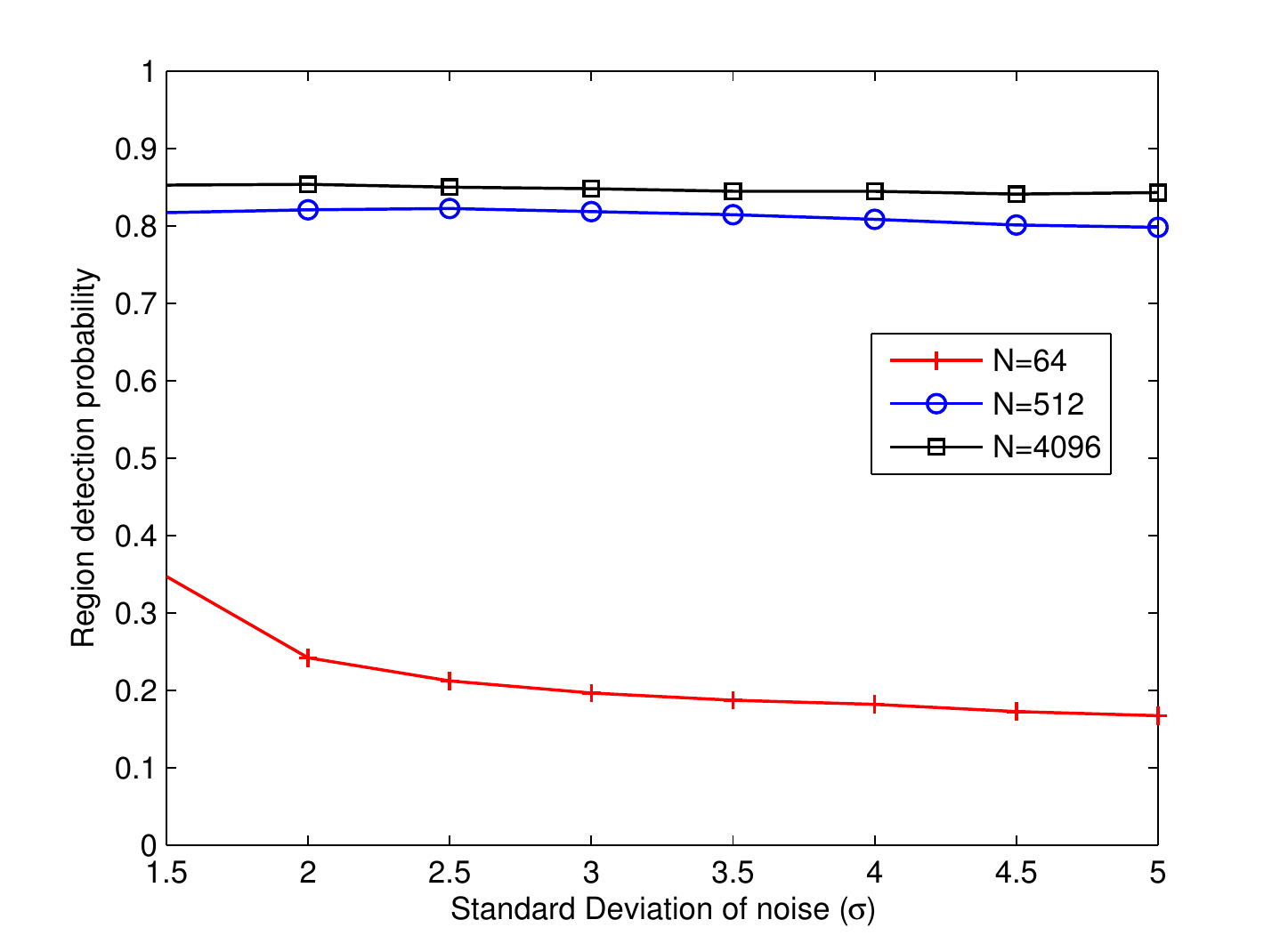}
\caption{Region detection probability versus the standard deviation of noise with varying number of sensors}
\label{Pd_sigma}
\end{figure}

\begin{table}[htb]
\caption{Target region detection probability for fixed noise variance ($\sigma =4$) with varying $N$ ($M=4$)}
\begin{center}
\begin{tabular}{|c|c|}
\hline
	$N$ & Target Region Detection probability \\
   \hline 
   \hline
   64 & 0.16753 \\ 
   \hline
   512 & 0.7982\\ 
   \hline
   4096 & 0.8433\\ 
   \hline
\end{tabular}
\end{center}
\label{table:pd}
\end{table}

\section{Localization in the presence of Byzantines}
\label{sec:byz}
Let us now consider the case when there are Byzantines in the network. As discussed before, Byzantines are local sensors which send false information to the FC to deteriorate the network's performance. We assume the presence of $B=\alpha N$ number of Byzantines in the network. In this paper, we assume that the Byzantines attack the network independently \cite{Vempaty_tsp} where the Byzantines flip their data with probability `1' before sending it to the FC. Note that the Byzantines can flip with any probability $\epsilon$. However, since it has been shown in \cite{Vempaty_tsp} that the optimal independent attack strategy for the Byzantines is to flip their data with probability `1', we focus on the optimal attack case which is $\epsilon=1$. In other words, the data sent by the $i^{th}$ sensor is given by:
\begin{equation}
\label{byz}
u_i=
\begin{cases}
D_i & \text{if $i^{th}$ sensor is honest} \\
\bar{D}_i & \text{if $i^{th}$ sensor is Byzantine}
\end{cases}.
\end{equation}

For such a system, it has been shown in \cite{Vempaty_tsp} that the FC becomes `blind' to the network's information for $\alpha \geq 0.5$. Therefore, for the remainder of the paper, we analyze the system when $\alpha < 0.5$. For the basic coding scheme described in Section \ref{basic}, each column in $C^k$ contains only one `1' and every row of $C^{k}$ contains exactly $\frac{N}{M^{k+1}}$ `1's. Therefore, the minimum Hamming distance of $C^k$ is $\frac{2N}{M^{k+1}}$ and, at the $(k+1)^{th}$ iteration, it can tolerate a total of at most $\frac{N}{M^{k+1}}-1$ faults (data falsification attacks) due to the presence of Byzantines in the network. This value is not very high and we would like to extend the basic scheme to a scheme which can handle more Byzantine faults.

\subsection{Exclusion Method with Weighted Average}
\label{exclusion}
As shown above, the scheme proposed in Section \ref{basic} has a Byzantine fault tolerance capability which is not very high. The performance can be improved by using an exclusion method for decoding where the two best regions are kept for next iteration and a weighted average is used to estimate the target location at the final step. This scheme builds on the basic coding scheme proposed in Section \ref{basic} with the following improvements:

\bi 
\item Since after every iteration two regions are kept, the code matrix after the $k^{th}$ iteration is of size $M \times \frac{2^kN}{M^k}$ and the number of iterations needed to stop the localization task needs to satisfy $k^{stop} < \log_{M/2}N$. 
\item At the final step, instead of taking an average of the sensor locations of the sensors present in the ROI at the final step, we take a weighted average of the sensor locations where the weights are the 1-bit decisions sent by these sensors. Since a decision $u_i=1$ would imply that the target is closer to the sensor $i$, a weighted average ensures that the average is taken only over the sensors for which the target is reported to be close. 
\ei

Therefore, the target location estimate is given by 
\begin{eqnarray}
\hat{\theta}_x=\frac{\sum_{i\in ROI_{k^{stop}}}{u_i x_i}}{\sum_{i\in ROI_{k^{stop}}}{u_i}}\\\ \text{and}\quad
\hat{\theta}_y=\frac{\sum_{i\in ROI_{k^{stop}}}{u_i y_i}}{\sum_{i\in ROI_{k^{stop}}}{u_i}}.
\end{eqnarray}

One can extend this scheme to consider other weights such as based on Euclidean distance which can be determined after processing the initial data to derive a coarse estimate of the target location. However, further processing is required for this and, therefore, we have not used such a scheme.
The exclusion method results in a better performance compared to the basic coding scheme since it keeps the two best regions after every iteration. This observation is also evident in the numerical results presented in Section~\ref{exclusive_res}. 

\subsection{Performance analysis}
\label{perf}

\subsubsection*{Byzantine Fault Tolerance Capability}
\label{fault}
When the exclusion based scheme described in Section \ref{exclusion} is used, since the two best regions are considered after every iteration, the fault tolerance performance improves and we can tolerate a total of at most $\frac{2^{k+1}N}{M^{k+1}}-1$ faults. This improvement in the fault tolerance capability can be observed in the simulation results presented in Section~\ref{exclusive_res}.

\begin{proposition}
\label{prop}
The maximum fraction of Byzantines that can be handled at the $(k+1)^{th}$ iteration by the proposed exclusion method based coding scheme is limited by $\alpha_f^k = \frac{2}{M}-\frac{M^k}{2^kN}$. 
\end{proposition}

\begin{proof}
The proof is straight forward and follows from the fact that the error correcting capability of the code matrix $C^k$ at $(k+1)^{th}$ iteration is at most $\frac{2^{k+1}N}{M^{k+1}}-1$. Since there are $\frac{2^kN}{M^k}$ sensors present during this iteration, the fraction of Byzantine sensors that can be handled is given by $\alpha_f^k = \frac{2}{M}-\frac{M^k}{2^kN}$.
\end{proof}

The performance bounds on the basic coding scheme presented in Section \ref{det} can be extended to the exclusion based coding scheme presented in Section \ref{exclusion}. We skip the details for the sake of brevity of the paper. When there are Byzantines in the network, the probabilities $q_{i,j}^k$ of \eqref{eq:prob1} become
\begin{eqnarray}
&&q_{i,j}^k=\nn\\
&&1-\left[(1-\alpha)\bar{F}\left(\eta_i^k-a_i;\sigma^2\right)+\alpha\left(1-\bar{F}\left(\eta_i^k-a_i;\sigma^2\right)\right)\right].\nn
\end{eqnarray}

We have shown in Section \ref{det} that the detection probability at every iteration approaches `1' as the number of sensors $N$ goes to infinity. However, this result only holds when the condition in \eqref{condition} is satisfied. Notice that, in the presence of Byzantines, we have

\begin{eqnarray}
&&q_{i,j}^k=\nn\\
&&\begin{cases}
(1-\alpha)\left(1-\bar{F}\left(\eta_i^k-a_i;\sigma^2\right)\right)+\alpha \bar{F}\left(\eta_i^k-a_i;\sigma^2\right) , \nn\\\qquad\qquad\qquad\qquad\qquad\qquad\qquad\qquad\text{for $i \in S_j^k$}\\
(1-\alpha)\bar{F}\left(\eta_i^k-a_i;\sigma^2\right) + \alpha \left(1-\bar{F}\left(\eta_i^k-a_i;\sigma^2\right)\right), \nn\\\qquad\qquad\qquad\qquad\qquad\qquad\qquad\qquad\text{for $i \in S_l^k$}
\end{cases},
\end{eqnarray}
which can be simplified as

\begin{equation}
q_{i,j}^k=
\begin{cases}
(1-\alpha)-(1-2\alpha)\bar{F}\left(\eta_i^k-a_i;\sigma^2\right) , &\text{for $i \in S_j^k$}\\
\alpha+(1-2\alpha)\bar{F}\left(\eta_i^k-a_i;\sigma^2\right) , &\text{for $i \in S_l^k$}
\end{cases}.\label{eq:cases_q}
\end{equation}
Now using the pairwise sum approach discussed in Section \ref{det}, we can re-write \eqref{eq:pairwise_sum} as follows:
\begin{eqnarray}
&&q_{i_j,j}^k+q_{i_l,j}^k=\nn\\
&&1-(1-2\alpha)\left[\bar{F}\left(\eta-a_{i_j};\sigma^2\right)-\bar{F}\left(\eta-a_{i_l};\sigma^2\right)\right],\label{eq:pairwise_byz}
\end{eqnarray}
which is an increasing function of $\alpha$ since $\bar{F}\left(\eta-a_{i_j};\sigma^2\right)>\bar{F}\left(\eta-a_{i_l};\sigma^2\right)$ for all finite $\sigma$ as discussed before. Therefore, when $\alpha < 0.5$, the pairwise sum in \eqref{eq:pairwise_byz} is strictly less than 1 and the condition \eqref{condition} is satisfied. However, when $\alpha\geq 0.5$, $\sum_{i\in S_j^k\cup S_l^k}q_{i,j}^k\geq \frac{N_k}{M}$. Therefore, the condition fails when $\alpha \geq 0.5$. It has been shown in \cite{Vempaty_tsp} that the FC becomes `blind' to the local sensor's information when $\alpha \geq 0.5$. Next we state the theorem when there are Byzantines in the network.
\begin{theorem}
Let $\alpha$ be the fraction of Byzantines in the networks. Under Assumption~\eqref{assumption}, when $\alpha<0.5$, $\displaystyle\lim_{N\rightarrow\infty} P_D=1$.
\end{theorem}

Note that the performance bounds derived can be used for system design. Let us consider $N$ sensors uniformly deployed in a square region. Let this region be split into $M$ equal regions. From Proposition \ref{prop}, we know that $\alpha_f^k$ is a function of $M$ and $N$. Also, the detection probability equations and bounds derived in Section \ref{det} are functions of $M$ and $N$. Hence, for given fault tolerance capability and region detection probability requirements, we can find the corresponding number of sensors ($N_{req}$) to be used and the number of regions to be considered at each iteration ($M_{req}$). We now present guidelines for system design of a network which adopts the proposed approach. Let us suppose that we need to design a system such that we split into $M=4$ regions after every iteration. How should a system designer decide the number of sensors $N$ in order to meet the target region detection probability and Byzantine fault tolerance capability requirements? Table~\ref{table:design} shows the performance of the system in terms of the target region detection probability and Byzantine fault tolerance capability with varying number of sensors found using the expressions derived in Proposition~\ref{prop} and in Section~\ref{det}.

\begin{table}[htb]
\caption{Target region detection probability and Byzantine fault tolerance capability with varying $N$ ($M=4$)}
\begin{center}
\begin{tabular}{|p{.5in}|p{1.25in}|p{1.25in}|}
\hline
	$N$ & Target Region Detection probability  & Byzantine fault tolerance capability\\
   \hline 
   \hline
   32 & 0.4253 & 0.4688\\ 
   \hline
   128 & 0.6817 & 0.4844\\ 
   \hline
   512 & 0.6994 & 0.4922\\ 
   \hline
\end{tabular}
\end{center}
\label{table:design}
\end{table}

From Table~\ref{table:design}, we can observe that the performance improves with increasing number of sensors. However, as a system designer, we would like to minimize the number of sensors that need to be deployed while assuring a minimum performance guarantee. In this example, if we are interested in achieving a region detection probability of approximately 0.7 and a Byzantine fault tolerance capability close to 0.5, we get $N=512$ sensors to be sufficient.

\subsection{Simulation Results}
\label{exclusive_res}
In this section, we present the simulation results to evaluate the performance of the proposed schemes in the presence of Byzantine faults. We analyze the performance using two performance metrics: mean square error (MSE) of the estimated location and probability of detection ($P_D$) of the target region. We use a network of $N=512$ sensors deployed in a regular $8 \times 8$ grid as shown in Fig. \ref{model_split}. Let $\alpha$ denote the fraction of Byzantines in the network that are randomly distributed over the network. The received signal amplitude at the local sensors is corrupted by AWGN noise with standard deviation $\sigma=3$. The power at the reference distance is $P_0=200$. At every iteration, the ROI is split into $M=4$ equal regions as shown in Fig. \ref{model_split}. We stop the iterations for the basic coding scheme after $k^{stop}= 2$ iterations. The number of sensors in the ROI at the final step are, therefore, $32$. In order to have a fair comparison, we stop the exclusion method after $k^{stop}=4$ iterations, so that there are again $32$ sensors in the ROI at the final step.

Fig. \ref{MSE} shows the performance of the proposed schemes in terms of the MSE of the estimated target location when compared with the traditional maximum likelihood estimation described by \eqref{MLE}. The MSE has been found by performing $1 \times 10^3$ Monte Carlo runs with the true target location randomly chosen in the $8 \times 8$ grid. 

\begin{figure}[htb]
\centering
\includegraphics[width = 3.5in,height=!]{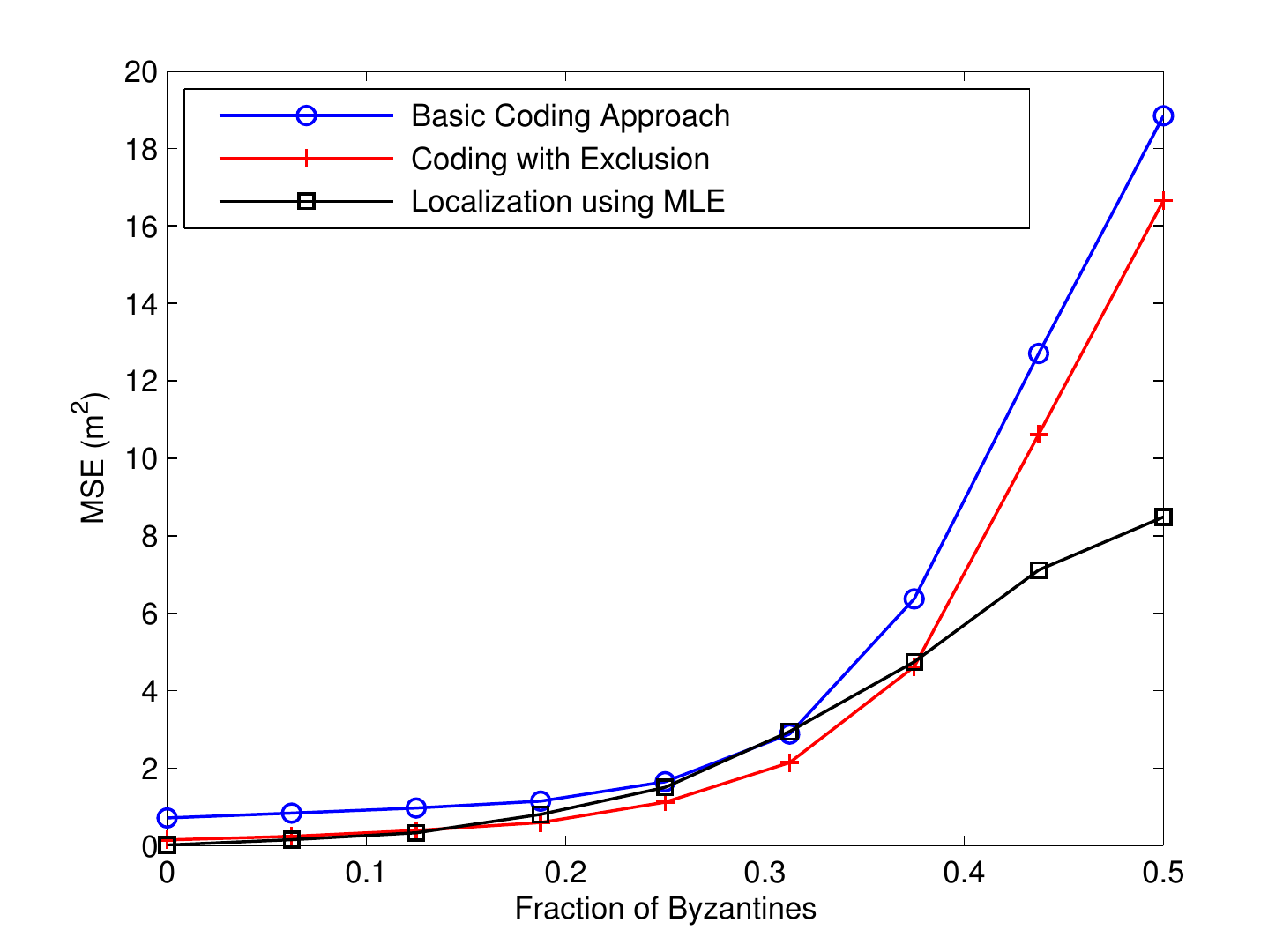}
\caption{MSE comparison of the three localization schemes}
\label{MSE}
\end{figure}

As can be seen from Fig. \ref{MSE}, the performance of the exclusion method based coding scheme is better than the basic coding scheme and outperforms the traditional MLE based scheme when $\alpha \le 0.375$. When $\alpha>0.375$  the traditional MLE based scheme has the best performance. 
However, it is important to note that the proposed schemes provide a coarse estimate as against the traditional MLE based scheme which optimizes over the entire ROI. Also, the traditional scheme is computationally much more expensive than the proposed coding based schemes. In the simulations performed, the proposed schemes are around 150 times faster than the conventional scheme when the global optimization toolbox in MATLAB was used for the optimization in ML based scheme. The computation time is very important in a scenario when the target is moving and a coarse location estimate is needed in a timely manner.

Fig. \ref{Pd_3} shows the performance of the proposed schemes in terms of the detection probability of the target region. The detection probability has been found by performing $1 \times 10^4$ Monte Carlo runs with the true target randomly chosen in the ROI. Fig. \ref{Pd_3} shows the reduction in the detection probability with increase in $\alpha$ when more sensors are Byzantines sending false information to the FC.

\begin{figure}[htb]
\centering
\includegraphics[width = 3.5in,height=!]{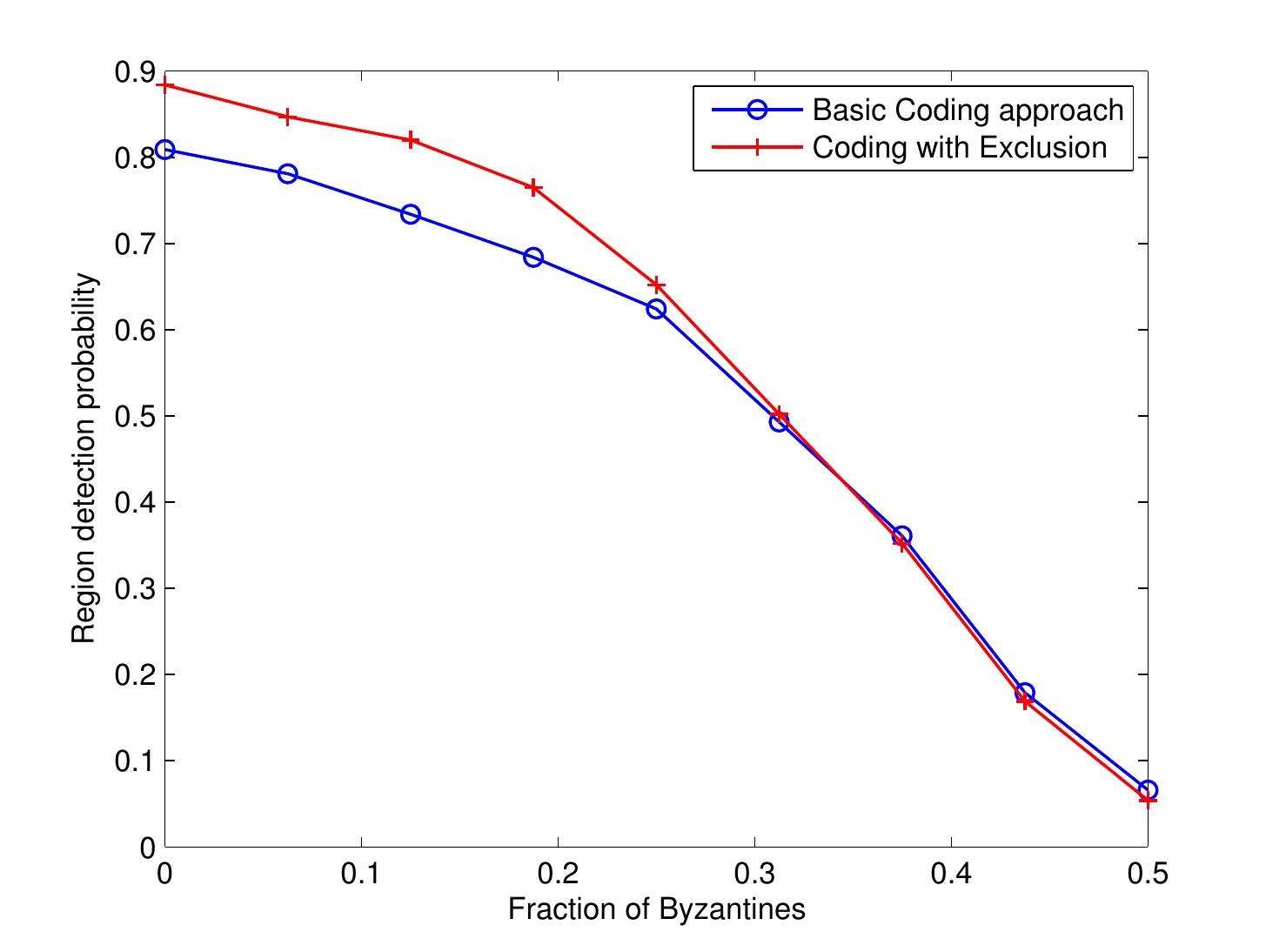}
\caption{Probability of detection of target region as a function of $\alpha$}
\label{Pd_3}
\end{figure}

In order to analyze the effect of the number of sensors on the performance, we perform simulations by changing the number of sensors and keeping the number of iterations the same as before. According to Proposition \ref{prop}, when $M=4$, the proposed scheme can asymptotically handle up to $50\%$ of the sensors being Byzantines. Figs. \ref{MSE_N} and \ref{Pd_N} show the effect of number of sensors on MSE and detection probability of the target region respectively when the exclusion method based coding scheme is used. As can be seen from both  figures (Figs. \ref{MSE_N} and \ref{Pd_N}), the fault-tolerance capability of the proposed scheme improves with increase in the number of sensors and approaches $\alpha_f^k=0.5$ asymptotically. Table \ref{table:pd2} shows the reduction of MSE with increasing $N$ for a fixed fraction of Byzantines, $\alpha$.

\begin{figure}[htb]
\centering
\includegraphics[width = 3.5in,height=!]{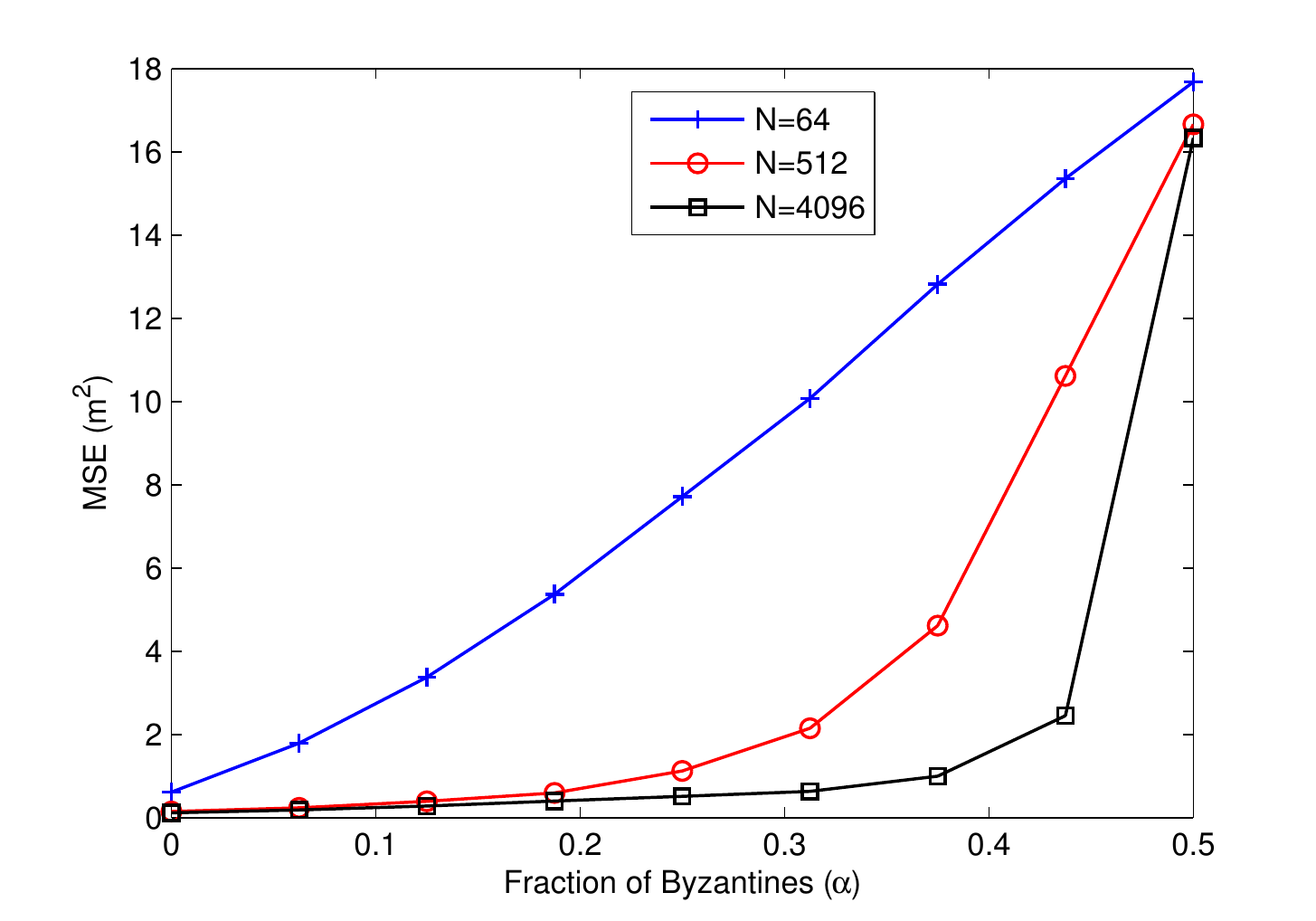}
\caption{MSE of the target location estimate with varying $N$}
\label{MSE_N}
\end{figure}

\begin{figure}[htb]
\centering
\includegraphics[width = 3.5in,height=!]{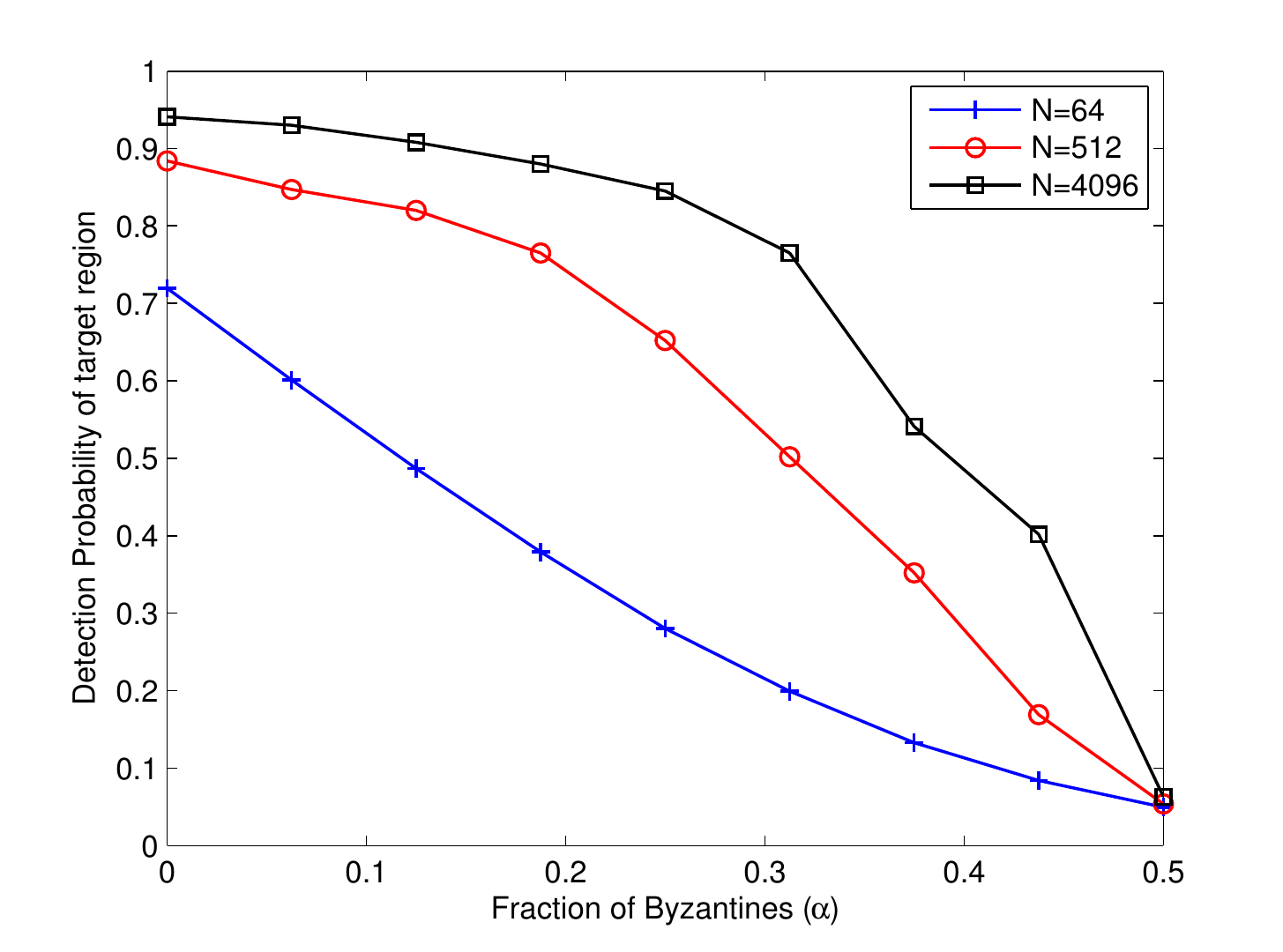}
\caption{Probability of detection of target region with varying $N$}
\label{Pd_N}
\end{figure}

\begin{table}[htb]
\caption{MSE of the target location estimate for fixed number of Byzantines ($\alpha =0.25$) with varying $N$}
\begin{center}
\begin{tabular}{|c|c|}
\hline
	$N$ & MSE ($m^2$) \\
   \hline 
   \hline
   64 & 7.79 \\ 
   \hline
   512 & 1.124\\ 
   \hline
   4096 & 0.5115\\ 
   \hline
\end{tabular}
\end{center}
\label{table:pd2}
\end{table}

\section{Soft-decision decoding for non-ideal channels}
\label{soft_dec}
In this section, we extend our scheme to counter the effect of non-ideal channels on system performance. Besides the faults due to the Byzantines in the network, the presence of non-ideal channels further degrades the localization performance. To combat the channel effects, we propose the use of a soft-decision decoding rule, at every iteration, instead of the minimum Hamming distance decoding rule. Note that the code design is independent of the hard-decoding or soft-decoding since according to the code-design, a sensor sends a `1' when the sensor decides that the target is in the same region as the sensor.

\subsection{Decoding rule}
\label{dec_rule}

At each iteration, the local sensors transmit their local decisions $\uu^k$ which are possibly corrupted due to the presence of Byzantines. Let the received analog data at the FC be represented as $\vv^k=[v_1^k,v_2^k,\cdots,v_{N_k}^k]$, where the received observations are related to the transmitted decisions as follows:

\begin{equation}
\label{eq:fading}
v_i^k=h_i^k(-1)^{u_i^k}\sqrt{E_b}+n_i^k, \qquad \text{$\forall i=\{1,\cdots,N_k\}$},
\end{equation}
where $h_i^k$ is the fading channel coefficient, $E_b$ is the energy per channel bit and $n_i^k$ is the additive white Gaussian noise with variance $\sigma^2_f$. In this paper, we assume the channel coefficients to be Rayleigh distributed with variance $\sigma_h^2$. 

We assume that the FC does not have knowledge of the fraction of Byzantines $\alpha$. Hence, instead of adopting the  reliability given in~\eqref{reliability-DCFECC}, we propose to use a simpler  reliability measure $\psi_i^k$ in our decoding rule that is not related to local decisions of sensors. It will be shown that this reliability measure performs well when there are Byzantines in the network. We define the reliability measure for each of the received bits as follows:

\begin{equation}
\label{eq:reliability}
\psi_i^k=\ln{\frac{P(v_i^k|u_i^k=0)}{P(v_i^k|u_i^k=1)}}
\end{equation}
for $i=\{1,\cdots,N\}$. Here $P(v_i^k|u_i^k)$ can be obtained from the statistical model of the Rayleigh fading channel considered in this paper. Define $F$-distance as $$d_F(\bpsi^k,\cc^k_{j+1})=\sum_{i=1}^{N_k}(\psi_i^k-(-1)^{c^k_{(j+1)i}})^2,$$ where $\bpsi^k=[\psi^k_1,\cdots,\psi^k_{N_k}]$ and $\cc^k_{j+1}$ is the $j^{th}$ row of the code matrix $C^k$. Then, the fusion rule is to decide the region $R_j^k$ for which the $F$-distance between $\bpsi^k$ and the row of $C^k$ corresponding to $R_j^k$ is minimized.

\subsection{Performance Analysis}
\label{sec:asymp_soft}
In this section, we present some bounds on the performance of the soft-decision decoding scheme in terms of the detection probability. Without loss of generality, we assume $E_b=1$. As mentioned before in \eqref{final_pd}, the overall detection probability is the product of the probability of detection at each iteration, $P_d^k$. We first present the following lemma without proof which is used to prove the theorem stated later in this section.

\begin{lemma}[\cite{Wang_twc06}]
\label{lemma:soft_dec}
Let $\tilde{\psi}_i^k=\psi_i^k-E[\psi_i^k|\theta]$, then

\begin{equation}
\label{eq:lemma}
E\left[(\tilde{\psi}_i^k)^2|\theta\right]\leq\frac{8}{\sigma^4}\left\{E[(h_i^k)^4]+E[(h_i^k)^2]\sigma^2_f\right\},
\end{equation}
\end{lemma}
where $\sigma^2$ is the variance of the noise at the local sensors whose observations follow \eqref{eq:measurement}. For the Rayleigh fading channel considered in this paper, both $E[(h_i^k)^4]$ and $E[(h_i^k)^2]$ are bounded and, therefore, the LHS of \eqref{eq:lemma} is also bounded.

\begin{lemma}
\label{lemma:soft_dec}
Let $\theta \in R_j^k$ be the fixed target location. Let $P_{e,j}^k(\theta)$ be the misclassification probability of the target region given $\theta \in R_j^k$ at the $(k+1)^{th}$ iteration. For the reliability vector $\bpsi^k=[\psi^k_1,\cdots,\psi^k_{N_k}]$ of the $N_k=N/M^k$ observations and code matrix $C^k$ used at the $(k+1)^{th}$ iteration, 
\begin{multline}
\label{eq:soft_dec_lemma}
P_{e,j}^k(\theta)\leq\sum_{0\leq l\leq M-1, l\neq j}\\
P\left\{\sum_{i\in S_j^k \cup S_l^k} Z_i^{jl}\tilde{\psi}_i^k\leq-\sum_{i\in S_j^k \cup S_l^k}Z_i^{jl}E[\psi_i^k|\theta]\bigg|\theta\right\},
\end{multline}
where $Z_i^{jl}=\frac{1}{2}((-1)^{c_{(j+1)i}^k}-(-1)^{c_{(l+1)i}^k})$.
\end{lemma}

\begin{proof}
\begin{eqnarray}
&&P_{e,j}^k(\theta)\nn\\
&=&P\{\text{detected region} \neq R_j^k|\theta\}\nonumber\\
&\leq&P\left\{d_F(\bpsi^k,\cc_{j+1}^k)\geq\min_{0\leq l\leq M-1, l\neq j} d_F(\bpsi^k,\cc_{l+1}^k)|\theta\right\}\nonumber\\
&\leq&\sum_{0\leq l\leq M-1, l\neq j}P\left\{d_F(\bpsi^k,\cc_{j+1}^k)\geq d_F(\bpsi^k,\cc_{l+1}^k)|\theta\right\}\nonumber\\
&=&\sum_{0\leq l\leq M-1, l\neq j}\nn\\
&&P\left\{\sum_{i=1}^{N_k}(\psi_i^k-(-1)^{c_{(j+1)i}^k})^2\geq (\psi_i^k-(-1)^{c_{(l+1)i}^k})^2|\theta\right\}\nonumber\\
&=&\sum_{0\leq l\leq M-1, l\neq j}P\left\{\sum_{i\in S_j^k \cup S_l^k}Z_i^{jl}\psi_i^k \leq0\bigg| \theta\right\}\label{eq:explain}\\
&=&\sum_{0\leq l\leq M-1, l\neq j}P\left\{\sum_{i\in S_j^k \cup S_l^k}Z_i^{jl}\tilde{\psi}_i^k \leq-\sum_{i\in S_j^k \cup S_l^k}Z_i^{jl} E[\psi_i^k|\theta]\bigg| \theta\right\}\nn,
\end{eqnarray}
where \eqref{eq:explain} comes from the fact that 
\begin{eqnarray}
&&(\psi_i^k-(-1)^{c_{(j+1)i}^k})^2- (\psi_i^k-(-1)^{c_{(l+1)i}^k})^2 \geq 0\nonumber\\
&\iff&-2((-1)^{c_{(j+1)i}^k}-(-1)^{c_{(l+1)i}^k})\psi_i^k\geq 0\nonumber\\
&\iff&Z_i^{jl}\psi_i^k \leq0\nonumber
\end{eqnarray}

\end{proof}

Let $\sigma_{\tilde{\psi}}^2(\theta)=\sum_{i\in S_j^k\cup S_l^k}E\left[(Z_i^{jl}\tilde{\psi}_i^k)^2|\theta\right]=\sum_{i\in S_j^k\cup S_l^k}E\left[(\tilde{\psi}_i^k)^2|\theta\right]$, then the above result can be re-written as 

\begin{multline}
\label{eq:tilde}
P_{e,j}^k(\theta) \leq \sum_{0\leq l\leq M-1, l\neq j}\\
P\left\{\frac{1}{\sigma_{\tilde{\psi}}(\theta)}\sum_{i\in S_j^k \cup S_l^k} Z_i^{jl}\tilde{\psi}_i^k<-\frac{1}{\sigma_{\tilde{\psi}}(\theta)}\sum_{i\in S_j^k \cup S_l^k}Z_i^{jl}E[\psi_i^k|\theta]\bigg|\theta\right\}.
\end{multline}
Under the assumption that $\frac{N}{M^{k+1}} \to \infty$ as $N \to \infty$ for $k=0,\cdots,k^{stop}$, we have the following result for asymptotic performance of the proposed soft-decision rule decoding based scheme.

\begin{theorem}
\label{asymp}
Under Assumption \eqref{assumption}, when $\alpha <0.5$,  $$\lim_{N \to \infty}P_D=1.$$
\end{theorem}

\begin{proof}
First we prove that when  $\alpha<0.5$, then
\begin{equation}
\label{eq:cond3}
\sum_{i\in S_j^k \cup S_l^k}Z_i^{jl}E[\psi_i^k|\theta] \to \infty,
\end{equation}
where $Z_i^{jl}=\frac{1}{2}((-1)^{c_{(j+1)i}^k}-(-1)^{c_{(l+1)i}^k})$.  Based on our code matrix design, $Z_i^{jl}$ for $i \in S_j^k\cup S_l^k$ is given as
\begin{equation}
Z_i^{jl}=
\begin{cases}
-1, &\text{for $i \in S_j^k$}\\
+1,&\text{for $i\in S_l^k$}
\end{cases}.
\end{equation} 

By using the pairwise summation approach discussed in Section \ref{det}, we notice that, for every sensor $i_j \in S_j^k$ and its corresponding sensor $i_l \in S_l^k$, when $\theta \in R_j^k$,
\begin{eqnarray}
Z_{i_j}^{jl}E[\psi_{i_j}^k|\theta]+Z_{i_l}^{jl}E[\psi_{i_l}^k|\theta]=E[(\psi_{i_l}^k-\psi_{i_j}^k)|\theta].
\end{eqnarray}

Now, for a given sensor $i$, we have the following,
\begin{eqnarray}
&&E[\psi_{i}^k|\theta]\nn\\
&=&P(u_i^k=0|\theta)E[\psi_{i}^k|\theta,u_i^k=0]\nn\\
&+&P(u_i^k=1|\theta)E[\psi_{i}^k|\theta,u_i^k=1]\\
&=&(1-P(u_i^k=1|\theta))E[\psi_{i}^k|u_i^k=0]\nn\\
&+&P(u_i^k=1|\theta)E[\psi_{i}^k|u_i^k=1]\\
&=&E[\psi_{i}^k|u_i^k=0]\nn\\
&+&P(u_i^k=1|\theta)\left[E[\psi_{i}^k|u_i^k=1]-E[\psi_{i}^k|u_i^k=0]\right],
\end{eqnarray}
where we used the facts that $P(u_i^k=0|\theta)+P(u_i^k=1|\theta)=1$ and that the value of $\psi_i^k$ depends only on $u_i^k$.

Note that the channel statistics are  the same for both the sensors. Therefore, $E[\psi_{i}^k|u_i^k=d]$ for $d=\{0,1\}$ given by
\begin{eqnarray}
E[\psi_{i}^k|u_i^k=d]=E\left[\ln\frac{P(v_{i}^k|u_{i}^k=0)}{P(v_{i}^k|u_{i}^k=1}\Bigg|u_i^k=d\right]\nn
\end{eqnarray}
is the same for both the sensors.

The pairwise sum $E[(\psi_{i_l}^k-\psi_{i_j}^k)|\theta]$ now simplifies to the following,
\begin{eqnarray}
&&E[(\psi_{i_l}^k-\psi_{i_j}^k)|\theta]\nn\\
&=&E[\psi_{i}^k|u_{i}^k=0]\nn\\
&+&P(u_{i_l}^k=1|\theta)\left[E[\psi_{i}^k|u_{i}^k=1]-E[\psi_{i}^k|u_{i}^k=0]\right]\nn\\
&-&E[\psi_{i}^k|u_{i}^k=0]\nn\\
&-&P(u_{i_j}^k=1|\theta)\left[E[\psi_{i}^k|u_{i}^k=1]-E[\psi_{i}^k|u_{i}^k=0]\right]\nn\\
&=&\left(P(u_{i_l}^k=1|\theta)-P(u_{i_j}^k=1|\theta)\right)\nn\\
&&\left[E[\psi_{i}^k|u_{i}^k=1]-E[\psi_{i}^k|u_{i}^k=0]\right].\label{eq:soft_dec_pair}
\end{eqnarray}
When $\theta\in R_j^k$, we have
\begin{eqnarray}
P(u_{i_j}^k=1|\theta)=\alpha+(1-2\alpha)\bar{F}\left(\eta-a_{i_j}\right)\\
P(u_{i_l}^k=1|\theta)=\alpha+(1-2\alpha)\bar{F}\left(\eta-a_{i_l}\right)
\end{eqnarray}
since the thresholds corresponding to sensors $i_j$ and $i_l$ are same due to Assumption \ref{assumption}. Therefore, 
\begin{equation}
P(u_{i_l}^k=1|\theta)-P(u_{i_j}^k=1|\theta)=(1-2\alpha)\left(\bar{F}\left(\eta-a_{i_l}\right)-\bar{F}\left(\eta-a_{i_j}\right)\right).\label{eq:soft_dec_alpha}
\end{equation}
Note that, since $\theta\in R_j^k$, $\bar{F}\left(\eta-a_{i_l}\right)<\bar{F}\left(\eta-a_{i_j}\right)$. Next we prove that 
\begin{equation}
E[\psi_{i}^k|u_{i}^k=1]-E[\psi_{i}^k|u_{i}^k=0]<0\label{E-E}
\end{equation}
for all finite noise variance of the fading channel ($\sigma_f^2$).

\begin{eqnarray}
&&
E[\psi_{i}^k|u_{i}^k=1]-E[\psi_{i}^k|u_{i}^k=0]\nonumber\\
&=&E\left[\ln\frac{P(v_{i}^k|u_{i}^k=0)}{P(v_{i}^k|u_{i}^k=1)}\Bigg|u_i^k=1\right]\nn\\
&-&E\left[\ln\frac{P(v_{i}^k|u_{i}^k=0)}{P(v_{i}^k|u_{i}^k=1)}\Bigg|u_i^k=0\right]\nn\\
&=&\int_{-\infty}^{\infty}P(v_{i}^k|u_{i}^k=1)\ln\frac{P(v_{i}^k|u_{i}^k=0)}{P(v_{i}^k|u_{i}^k=1)}\ dv_i^k\nn\\
&-&\int_{-\infty}^{\infty}P(v_{i}^k|u_{i}^k=0)\ln\frac{P(v_{i}^k|u_{i}^k=0)}{P(v_{i}^k|u_{i}^k=1)}\ dv_i^k\nn\\
&=&-D(P(v_{i}^k|u_{i}^k=1)||P(v_{i}^k|u_{i}^k=0))\nn\\
&&-D(P(v_{i}^k|u_{i}^k=0)||P(v_{i}^k|u_{i}^k=1)),
\end{eqnarray}
where $D(p||q)$ is the Kullback-Leiber distance between probability distributions $p$ and $q$. Since $P(v_{i}^k|u_{i}^k=1)\neq P(v_{i}^k|u_{i}^k=0)$ for all finite $\sigma_f^2$, we have
$D(P(v_{i}^k|u_{i}^k=1)||P(v_{i}^k|u_{i}^k=0))>0$ and $D(P(v_{i}^k|u_{i}^k=0)||P(v_{i}^k|u_{i}^k=1))>0$. This concludes that $E[\psi_{i}^k|u_{i}^k=1]-E[\psi_{i}^k|u_{i}^k=0]<0$. Hence, when $\alpha <1/2$, from \eqref{eq:soft_dec_pair}, \eqref{eq:soft_dec_alpha}, and \eqref{E-E}, $E[(\psi_{i_l}^k-\psi_{i_j}^k)|\theta]>0$ and the condition $\sum_{i\in S_j^k \cup S_l^k}Z_i^{jl}E[\psi_i^k|\theta] \to \infty$ is satisfied. 

We now show that when the condition \eqref{eq:cond3} is satisfied, the proposed scheme asymptotically attains perfect detection probability. 
\begin{eqnarray}
&&\lim_{N\to\infty}P_D\nn\\
&=&\lim_{N\to\infty}\prod_{k=0}^{k^{stop}}P_d^k\nn\\
&=&\prod_{k=0}^{k^{stop}}\lim_{N\to\infty}\Bigg[1-\nn\\
&&\sum_{j=0}^{M-1}P\left\{\theta\in R_j^k\right\}P\left\{\text{detected region}\neq R_j^k|\theta \in R_j^k\right\}\Bigg]\nn\\
&=&\prod_{k=0}^{k^{stop}}\lim_{N\to\infty}\Bigg[1-\frac{1}{M}\sum_{j=0}^{M-1}\nn\\
&&\int_{\theta}P\left\{\theta|\theta\in R_j^k\right\}P\left\{\text{detected region}\neq R_j^k|\theta,\theta \in R_j^k\right\}d\theta\Bigg].\nn\\
\end{eqnarray}
Define 
\begin{equation}
P_{e,j,\text{max}}^k\defeq \max_{\theta\in R_j^k}P_{e,j}^k(\theta)
\end{equation}
and 
\begin{equation}
P_{e,\text{max}}^k\defeq \max_{0\leq j\leq M-1}P_{e,j,\text{max}}^k.
\end{equation}

Then,
\begin{eqnarray}
&&\lim_{N\to\infty}P_D\nn\\
&=&\prod_{k=0}^{k^{stop}}\lim_{N\to\infty}\Bigg[1-\frac{1}{M}\sum_{j=0}^{M-1}\int_{\theta}P\left\{\theta|\theta\in R_j^k\right\}P_{e,j}^k(\theta) d\theta\Bigg]\nn\\
&\geq&\prod_{k=0}^{k^{stop}}\lim_{N\to\infty}\Bigg[1-\frac{1}{M}\nn\\
&&\sum_{j=0}^{M-1}\int_{\theta\in R_j^k}P\left\{\theta|\theta\in R_j^k\right\}P_{e,j,\text{max}}^kd\theta\Bigg]\nn\\
&=&\prod_{k=0}^{k^{stop}}\lim_{N\to\infty}\Bigg[1-\nn\\
&&\frac{1}{M}\sum_{j=0}^{M-1}P_{e,j,\text{max}}^k\int_{\theta\in R_j^k}P\left\{\theta|\theta\in R_j^k\right\}d\theta\Bigg]\nn\\
&\geq&\prod_{k=0}^{k^{stop}}\lim_{N\to\infty}\Bigg[1-\frac{P_{e,\text{max}}^k}{M}\sum_{j=0}^{M-1}1\Bigg]\nn\\
&=&\prod_{k=0}^{k^{stop}}\Bigg[1-\lim_{N\to\infty}P_{e,\text{max}}^k\Bigg].\label{eq:lim}
\end{eqnarray}

Since $E\left[(\tilde{\psi}_i^k)^2|\theta\right]$ is bounded as shown by Lemma \ref{lemma:soft_dec}, Lindeberg condition \cite{feller} holds and $\frac{1}{\sigma_{\tilde{\psi}}(\theta)}\sum_{i\in S_j^k \cup S_l^k} Z_i^{jl}\tilde{\psi}_i^k$ tends to a standard Gaussian random variable by Lindeberg central limit theorem \cite{feller}.  Therefore, from \eqref{eq:tilde}, we have
\begin{eqnarray}
&&\lim_{N\to\infty}P_{e,j}^k(\theta)\nn\\
&\leq&\lim_{N\to\infty}\sum_{0\leq l\leq M-1, l\neq j}P\Bigg\{\frac{1}{\sigma_{\tilde{\psi}}(\theta)}\sum_{i\in S_j^k \cup S_l^k} Z_i^{jl}\tilde{\psi}_i^k<\nn\\
&&-\frac{1}{\sigma_{\tilde{\psi}}(\theta)}\sum_{i\in S_j^k \cup S_l^k}Z_i^{jl}E[\psi_i^k|\theta]\bigg|\theta\Bigg\} \\
&=&\sum_{0\leq l\leq M-1, l\neq j}\lim_{N\to\infty}Q\left(\frac{1}{\sigma_{\tilde{\psi}}(\theta)}\sum_{i\in S_j^k \cup S_l^k}Z_i^{jl}E[\psi_i^k|H_j^k]\right).\nn
\end{eqnarray}

Since, for a fixed $\theta$, $\sigma_{\tilde{\psi}}(\theta)$ will grow slower than $\sum_{i\in S_j^k \cup S_l^k}Z_i^{jl}E[\psi_i^k|\theta]$ when $\sum_{i\in S_j^k \cup S_l^k}Z_i^{jl}E[\psi_i^k|\theta] \to \infty$, $\lim_{N\to \infty}P_{e,j}^k(\theta) = 0$ for all $\theta$. Hence, $\lim_{N\to\infty}P_{e,\text{max}}^k = 0$ and from \eqref{eq:lim}, $\lim_{N\to\infty}P_D = 1$ for all finite noise variance.
\end{proof} 

Note that the detection probability of the proposed scheme can approach `1' even for extremely bad channels with very low channel capacity. This is true because, when $M$ increases sub-linearly with $N$, i.e., when $\frac{N}{M^{k+1}} \to \infty$ as $N \to \infty$ for $k=0,\cdots,k^{stop}$, as $N$ approaches infinity, the code rate of the code matrix approaches zero. Hence, even for extremely bad channels, the code rate is still less than the channel capacity.

\subsection{Numerical Results}
In this section, we present some numerical results which show the improvement in the system performance when soft-decision decoding rule is used instead of the hard-decision decoding rule in the presence of Byzantines and non-ideal channels. As defined before, $\alpha$ represents the fraction of Byzantines and we evaluate the performance of the basic coding approach with soft-decision decoding at the FC. We simulate the scenario with following system parameters: $N=512$, $M=4$, $A=8^2=64$~sq.~units, $P_0=200$, local sensor observations are corrupted with Gaussian noise with $\sigma=3$, $E_b=1$, $\sigma_f=3$ and $E[(h_i^k)^2]=1$ which corresponds to $\sigma_h^2=1-\frac{\pi}{4}$. The basic coding approach is stopped after $k^{stop}=2$ iterations. Note that in the presence of non-ideal channels, $\alpha_{blind}$ is less than $0.5$ since the non-ideal channels add to the errors at the FC. The number of Byzantine faults which the network can handle reduces and is now less than $0.5$. In our simulations, we observe that the performance of the schemes completely deteriorates when $\alpha \to 0.4$ (as opposed to $0.5$ observed before) and, therefore, we plot the results for the case when $\alpha \leq 0.4$

Fig. \ref{MSE_SD} shows the reduction in mean square error when the soft-decision decoding rule is used instead of the hard-decision decoding rule. Similarly, Fig. \ref{Pd_SD} shows the improvement in target region detection probability when the soft-decision decoding rule is used. The plots are for $5 \times 10^3$ Monte-Carlo simulations.

\begin{figure}[htb]
\centering
\includegraphics[width = 3.5in,height=!]{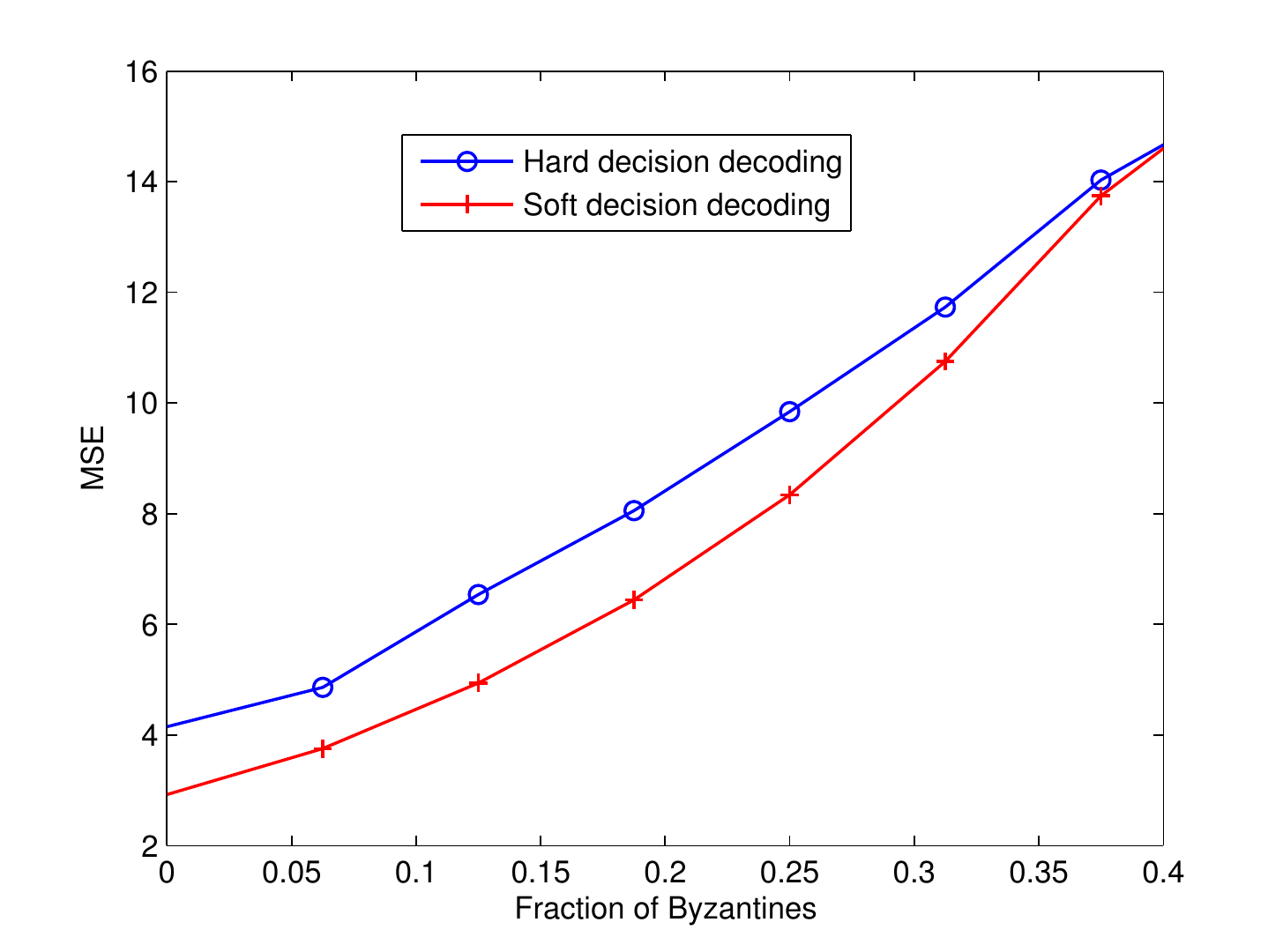}
\caption{MSE comparison of the basic coding scheme using soft- and hard- decision decoding}
\label{MSE_SD}
\end{figure}

\begin{figure}[htb]
\centering
\includegraphics[width = 3.5in,height=!]{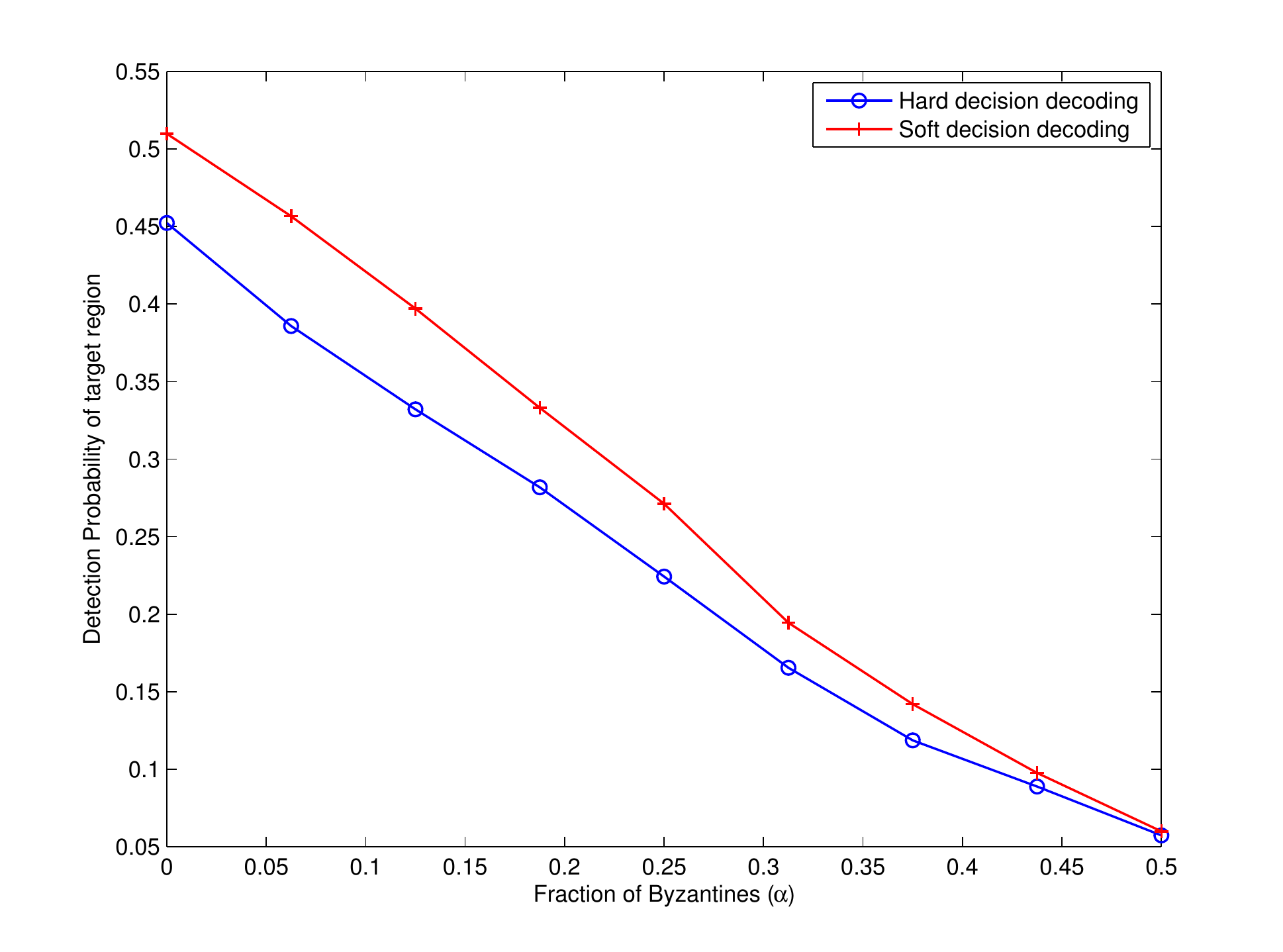}
\caption{Probability of detection of target region comparison of the basic coding scheme using soft- and hard- decision decoding}
\label{Pd_SD}
\end{figure}

As the figures suggest, the performance deteriorates in the presence of non-ideal channels. Also, the performance worsens with an increase in the number of Byzantines. The performance can be improved by using the exclusion method based coding approach as discussed in Section \ref{sec:byz} in which two regions are stored after every iteration. Figs. \ref{MSE_exclusive_SD} and \ref{Pd_exclusive_SD} show this improved performance as compared to the basic coding approach. Note that the exclusion method based coding approach also follows the same trend as the basic coding approach with soft-decision decoding performing better than hard-decision decoding.

\begin{figure}[htb]
\centering
\includegraphics[width = 3.5in,height=!]{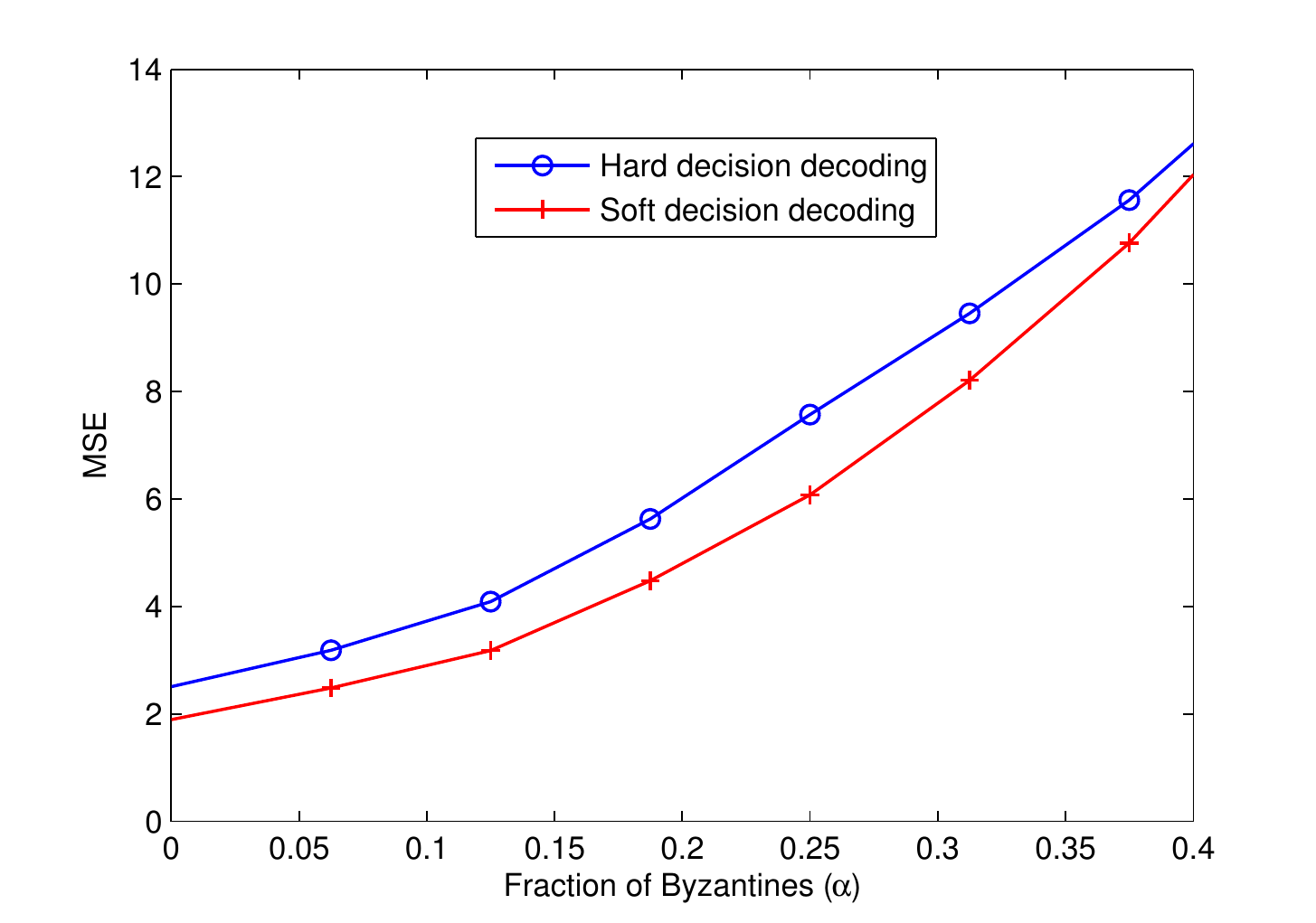}
\caption{MSE comparison of the exclusion coding scheme using soft- and hard- decision decoding}
\label{MSE_exclusive_SD}
\end{figure}

\begin{figure}[htb]
\centering
\includegraphics[width = 3.5in,height=!]{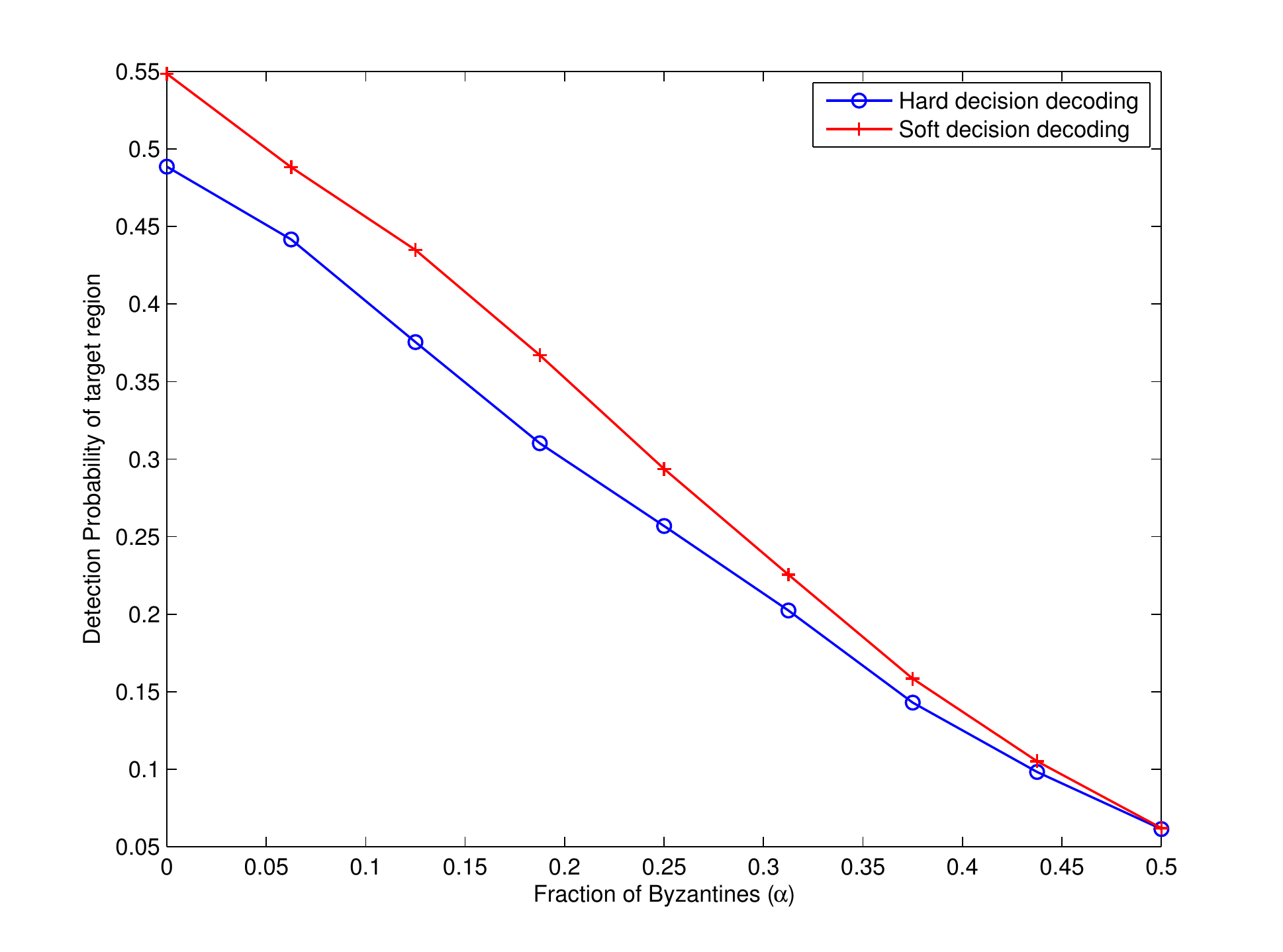}
\caption{Probability of detection of target region comparison of the exclusion coding scheme using soft- and hard- decision decoding}
\label{Pd_exclusive_SD}
\end{figure}

In our theoretical analysis, we have shown that the probability of region detection asymptotically approaches `1' irrespective of the finite noise variance. Fig.~\ref{fig:Pd_sig_f} presents this result that the region detection probability approaches `1' as the number of sensors approach infinity. Observe that for a fixed noise variance, the region detection probability increases with increase in the number of sensors and approaches `1' as $N \to \infty$. However, as $\sigma_f$ increases, the convergence rate decreases. For example, when $\sigma_f=1.5$, $N=4096$ is large enough to have $P_D$ close to 0.9. However, for $\sigma_f=4$, $N=4096$ results in $P_D=0.65$ which is not very large. It is expected that $P_D \to 1$ much later for $\sigma_f=4$ and, therefore, the convergence rate is less compared to when $\sigma_f=1.5$.

\begin{figure}[htb]
\centering
\includegraphics[width = 3.5in,height=!]{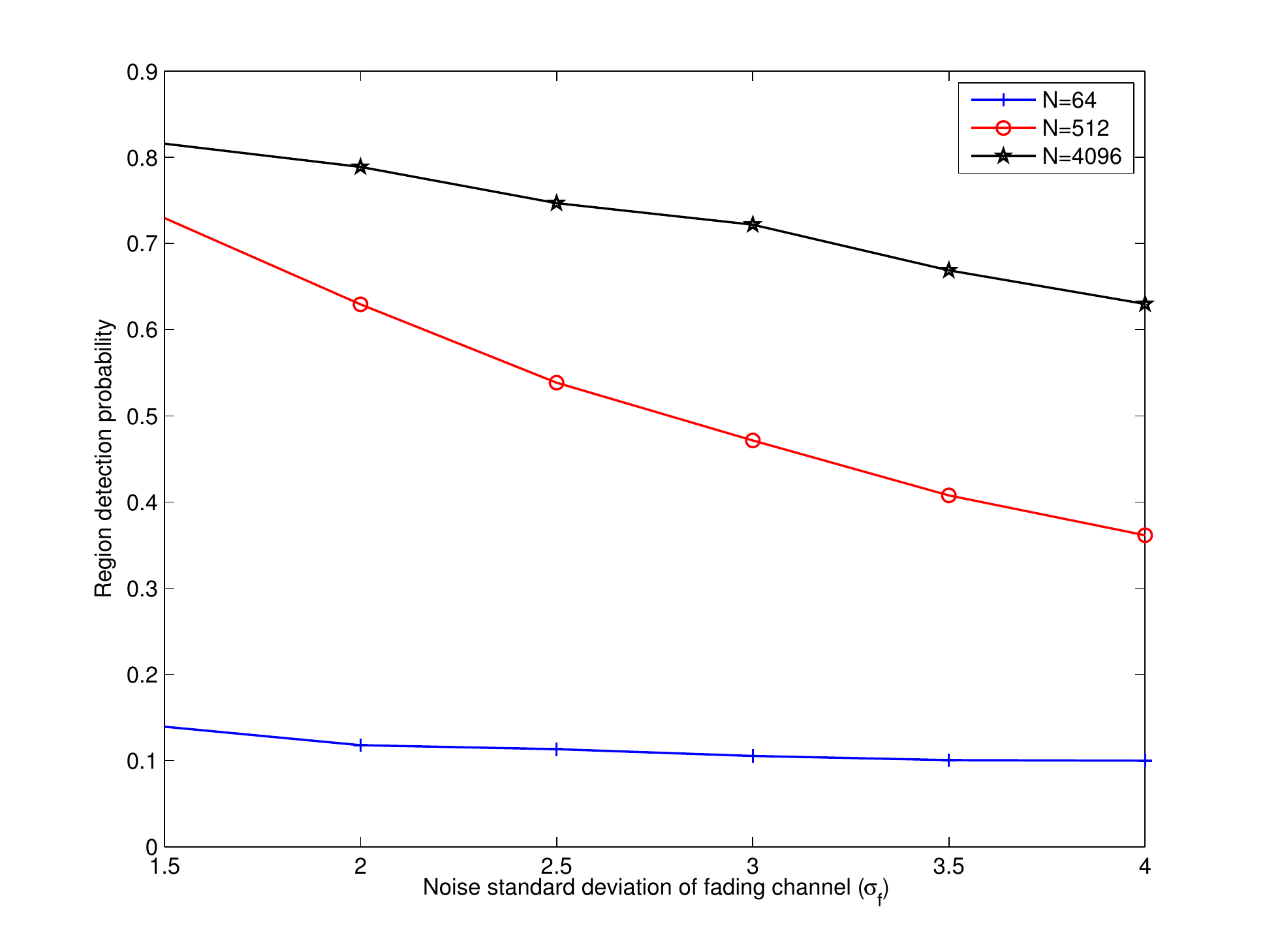}
\caption{Probability of detection of target region of the exclusive coding scheme using soft- decision decoding with varying number of sensors ($N$)}
\label{fig:Pd_sig_f}
\end{figure}

\section{Conclusion}
\label{conc}
In this paper, we considered the problem of target localization in wireless sensor networks. Traditionally, research has focused on conventional maximum likelihood approaches for estimating the target location. However, maximum likelihood based approaches are computationally very expensive. To reduce the computational complexity, we proposed a novel coding theory based technique for target localization. Modeling the estimation problem as an iterative classification problem, we can determine a coarse estimate of the target location in a computationally efficient manner. This efficiency in terms of computation becomes important in a scenario when the target is not stationary. The proposed scheme estimates the target location iteratively using $M$-ary classification at each iteration. We provided the theoretical analysis of the proposed scheme in terms of the detection probability of the target region. Considering the presence of Byzantines (malicious sensors) in the network, we modified our approach to increase the fault-tolerance capability of the coding scheme used. This approach, called the exclusion method based approach, is more tolerant to the presence of Byzantines than the basic coding scheme. We showed with simulations that the exclusion method based scheme provides an accurate estimate of the target location in a very efficient manner than the traditional MLE based scheme and also has a better Byzantine fault tolerance capability. We also considered the effect of non-ideal channels between local sensors and the fusion center. To minimize the effects of these non-ideal channels, we proposed soft-decision decoding at the fusion center. We showed with simulations, the improvement in performance of soft-decision decoding rule based scheme over hard-decision decoding rule based scheme in the presence of non-ideal channels. In the future, we plan to extend our work by relaxing Assumption \ref{assumption} and to also derive the convergence rates using Berry-Essen inequalities. One can also extend this work to the case of target tracking when the target's location is changing with time and the sensor network's aim is to track the target's motion. The proposed schemes provide an insight on $M$-ary search trees and show that the idea of coding based schemes can also be used for other signal processing applications. For example, the application involving `search' such as rumor source localization in social networks.


\begin{IEEEbiography}[{\includegraphics[width=1in, height=1.25in,clip,keepaspectratio]{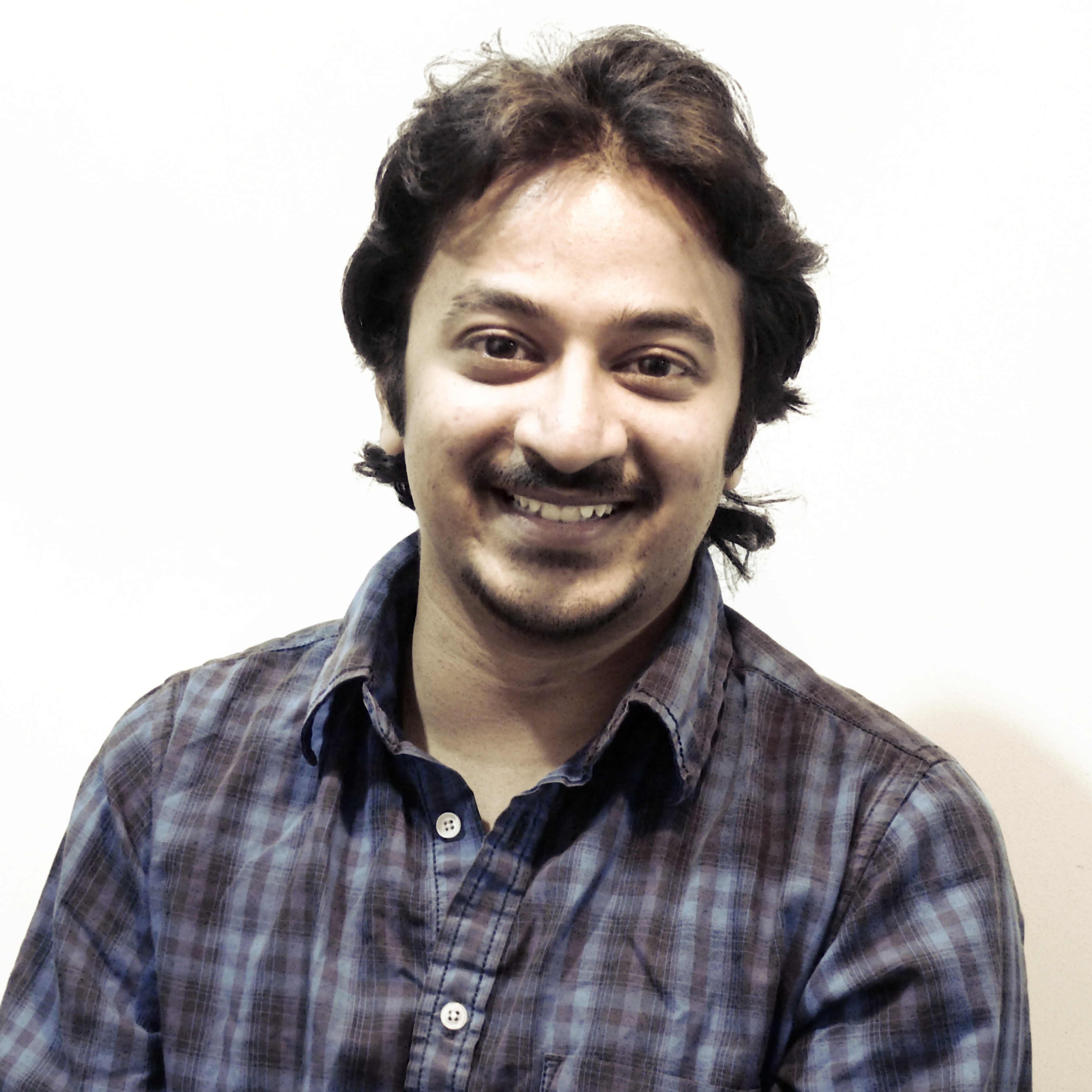}}]
{Aditya Vempaty}(S'12) was born in Hyderabad, India, on August 3, 1989. He received the B. Tech degree in electrical engineering from Indian Institute of Technology (IIT), Kanpur, India, in 2011 with academic excellence awards for consecutive years. Since 2011, he has been working towards his Ph.D. degree in electrical engineering at Syracuse University.

He is a Graduate Research Assistant at the Sensor Fusion Laboratory where he was also an Undergraduate Research Intern during summer of 2010. He was a Graduate Research Intern in the Data Systems Group at Intel Corporation in Santa Clara, CA in summer 2013. His research interests include statistical signal processing, human decision making, security in sensor networks, and data fusion. He has been awarded the Syracuse University Graduate Fellowship award for the academic year 2013-14. He is a member of Phi Kappa Phi and Golden Key International Honor Society. 

\end{IEEEbiography}
\begin{IEEEbiography}[{\includegraphics[width=1in,height=1.25in,clip,keepaspectratio]{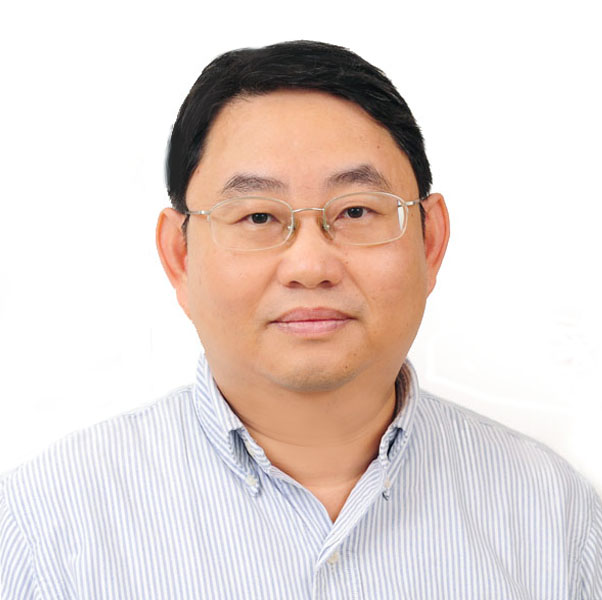}}]{Yunghsiang S.~Han}
(S'90--M'93--SM'08--F'11) was  born in Taipei, Taiwan, 1962. He received B.Sc. and M.Sc. degrees in electrical engineering from the
National Tsing Hua University, Hsinchu, Taiwan, in 1984 and 1986, respectively, and a Ph.D. degree from the School of Computer and
Information Science, Syracuse University, Syracuse, NY, in 1993.
 
He was from 1986 to 1988 a lecturer at Ming-Hsin Engineering College,
Hsinchu, Taiwan. He was a teaching assistant from 1989 to 1992, and a research associate in the School of Computer and Information
Science, Syracuse University from 1992 to 1993. He was, from 1993 to 1997, an Associate Professor in the Department of Electronic
Engineering at Hua Fan College of Humanities and Technology, Taipei Hsien, Taiwan. He was with the Department of Computer Science and
Information Engineering at National Chi Nan University, Nantou, Taiwan from 1997 to 2004. He was promoted to Professor in 1998. He
was a visiting scholar in the Department of Electrical Engineering at University of Hawaii at Manoa, HI from June to October 2001, 
the SUPRIA visiting research scholar in the Department of Electrical Engineering and Computer Science and CASE center at Syracuse
University, NY from September 2002 to January 2004 and July 2012 to June 2013, and the visiting scholar in the Department of Electrical 
and Computer Engineering at University of Texas at Austin, TX from August 2008 to June 2009. He was with the Graduate Institute of Communication Engineering at National
Taipei University, Taipei, Taiwan from August 2004 to July 2010. From August 2010, 
he is with the Department of Electrical Engineering at National Taiwan University of Science and Technology as Chair professor. 
His research interests are in error-control coding, wireless networks, and security.

Dr. Han was a winner of the 1994 Syracuse University Doctoral Prize and a Fellow of IEEE.
\end{IEEEbiography}

\begin{IEEEbiography}[{\includegraphics[width=1in, height=1.25in,clip,keepaspectratio]{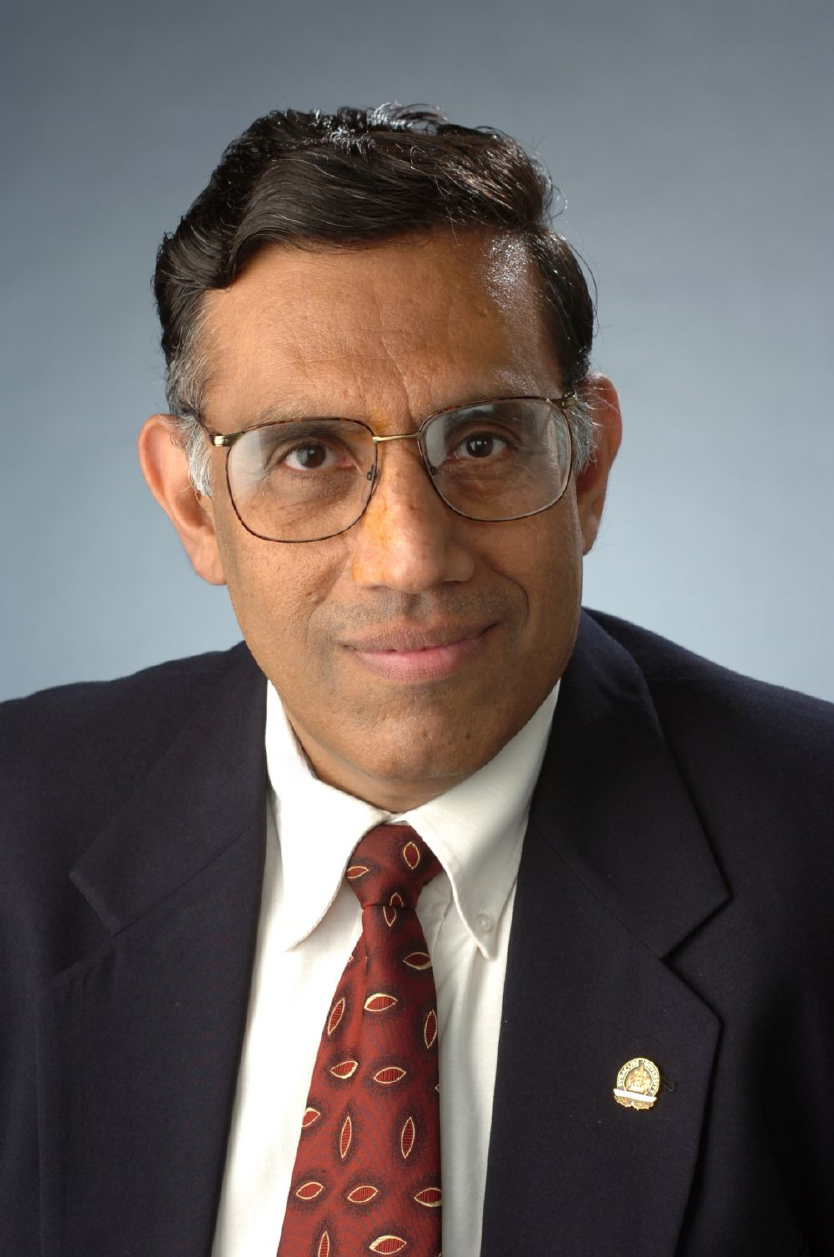}}]
{Pramod K. Varshney} (S'72-M'77-SM'82-F'97) was born in Allahabad, India, on July 1, 1952. He received the B.S. degree in electrical engineering and computer science (with highest honors), and the M.S. and Ph.D. degrees in electrical engineering from the University of Illinois at Urbana-Champaign in 1972, 1974, and 1976 respectively.

From 1972 to 1976, he held teaching and research assistantships with the University of Illinois. Since 1976, he has been with Syracuse University, Syracuse, NY, where he is currently a Distinguished
Professor of Electrical Engineering and Computer Science and the Director of CASE: Center for Advanced Systems and Engineering. He served as the Associate Chair of the department from 1993 to 1996. He is also an Adjunct Professor of Radiology at Upstate Medical University, Syracuse. His current research interests are in distributed sensor networks and data fusion, detection and estimation theory, wireless communications, image processing, radar signal processing, and remote sensing. He has published extensively. He is the author of Distributed Detection and Data Fusion (New York: Springer-Verlag, 1997). He has served as a consultant to several major companies.

Dr. Varshney was a James Scholar, a Bronze Tablet Senior, and a Fellow while at the University of Illinois. He is a member of Tau Beta Pi and is the recipient of the 1981 ASEE Dow Outstanding Young Faculty Award. He was elected to the grade of Fellow of the IEEE in 1997 for his contributions in the area of distributed detection and data fusion. He was the Guest Editor of the Special Issue on Data Fusion of the IEEE PROCEEDINGS January 1997. In 2000, he received the Third Millennium Medal from the IEEE and ChancellorÕs Citation for exceptional academic achievement at Syracuse University. He is the recipient of the IEEE 2012 Judith A. Resnik Award. He is on the Editorial Board of the Journal on Advances in Information Fusion. He was the President of International Society of Information Fusion during 2001.
\end{IEEEbiography}


\begin{thebibliography}{10}
\providecommand{\url}[1]{#1}
\csname url@samestyle\endcsname
\providecommand{\newblock}{\relax}
\providecommand{\bibinfo}[2]{#2}
\providecommand{\BIBentrySTDinterwordspacing}{\spaceskip=0pt\relax}
\providecommand{\BIBentryALTinterwordstretchfactor}{4}
\providecommand{\BIBentryALTinterwordspacing}{\spaceskip=\fontdimen2\font plus
\BIBentryALTinterwordstretchfactor\fontdimen3\font minus
  \fontdimen4\font\relax}
\providecommand{\BIBforeignlanguage}[2]{{%
\expandafter\ifx\csname l@#1\endcsname\relax
\typeout{** WARNING: IEEEtran.bst: No hyphenation pattern has been}%
\typeout{** loaded for the language `#1'. Using the pattern for}%
\typeout{** the default language instead.}%
\else
\language=\csname l@#1\endcsname
\fi
#2}}
\providecommand{\BIBdecl}{\relax}
\BIBdecl

\bibitem{akyildiz_commag02}
I.~F. Akyildiz, W.~Su, Y.~Sankarasubramaniam, and E.~Cayirci, ``A survey on
  sensor networks,'' \emph{IEEE Commun. Mag.}, vol.~40, no.~8, pp. 102--114,
  Aug. 2002.

\bibitem{niu_tsp06}
R.~Niu and P.~K. Varshney, ``Target location estimation in sensor networks with
  quantized data,'' \emph{IEEE Trans. Signal Process.}, vol.~54, no.~12, pp.
  4519--4528, Dec. 2006.

\bibitem{Wang_TMC12}
X.~Wang, M.~Fu, and H.~Zhang, ``Target tracking in wireless sensor networks
  based on the combination of \protect{KF} and \protect{MLE} using distance
  measurements,'' \emph{IEEE Trans. Mobile Comp.,}, vol.~11, no.~4, pp. 567
  --576, April 2012.

\bibitem{Veeravalli&Varshney_12}
V.~Veeravalli and P.~K. Varshney, ``Distributed inference in wireless sensor
  networks,'' \emph{Philosophical Transactions of the Royal Society A:
  Mathematical, Physical and Engineering Sciences}, vol. 370, no. 1958, pp.
  100--117, Jan 2012.

\bibitem{molnar_icassp01}
L.~M. Kaplan, L.~Qiang, and N.~Molnar, ``Maximum likelihood methods for
  bearings-only target localization,'' in \emph{Proc. IEEE Int. Conf. Acoust.,
  Speech, Signal Process. (ICASSP 2001)}, Salt Lake City, UT, May 2001.

\bibitem{Chen&etal:01Icassp}
J.~Chen, R.~Hudson, and K.~Yao, ``A maximum likelihood parametric approach to
  source localization,'' in \emph{Proc. IEEE Int. Conf. Acoust., Speech, Signal
  Process. (ICASSP 2001)}, vol.~5, Salt Lake City, UT, May 2001, pp.
  3013--3016.

\bibitem{hu_eurasip03}
D.~Li and Y.~H. Hu, ``Energy-based collaborative source localization using
  acoustic microsensor array,,'' \emph{EURASIP Journal on Applied Signal
  Processing}, vol. 2003, no.~4, pp. 321--337, 2003.

\bibitem{hu_tsp05}
X.~Sheng and Y.~H. Hu, ``Maximum likelihood multiple-source localization using
  acoustic energy measurements with wireless sensor networks,'' \emph{IEEE
  Trans. Signal Process.}, vol.~53, no.~1, pp. 44--53, Jan. 2005.

\bibitem{HKD06}
H.~K.~D. Sarma and A.~Kar, ``Security threats in wireless sensor networks,'' in
  \emph{Proc. 40th Annual IEEE International Carnahan Conferences Security
  Technology,}, Oct. 2006, pp. 243 --251.

\bibitem{Vempaty_SPM13}
A.~Vempaty, L.~Tong, and P.~K. Varshney, ``Distributed inference with
  {B}yzantine data: state-of-the-art review on data falsification attacks,''
  \emph{IEEE Signal Process. Mag.}, vol.~30, no.~5, pp. 65--75, Sep. 2013.

\bibitem{Vempaty_tsp}
A.~Vempaty, O.~Ozdemir, K.~Agrawal, H.~Chen, and P.~K. Varshney, ``Localization
  in wireless sensor networks: {B}yzantines and mitigation techniques,''
  \emph{IEEE Trans. Signal Process.}, vol.~61, no.~6, pp. 1495--1508, Mar. 15
  2013.

\bibitem{varshney_spmag06}
B.~Chen, L.~Tong, and P.~K. Varshney, ``Channel-aware distributed detection in
  wireless sensor networks,'' \emph{IEEE Signal Process. Mag. (Special Issue on
  Distributed Signal Processing for Sensor Networks)}, vol.~23, pp. 16--26,
  Jul. 2006.

\bibitem{Pottie&Kaiser:00ACM}
G.~Pottie and W.~Kaiser, ``Wireless integrated network sensors,''
  \emph{Communications of the ACM}, vol.~43, pp. 52--58, May 2000.

\bibitem{islam}
M.~R. Islam, ``Error correction codes in wireless sensor network: An energy
  aware approach,'' \emph{International Journal of Computer \& Information
  Engineering}, vol.~4, no.~1, pp. 59--64, Jan. 2010.

\bibitem{LIN83}
S.~Lin and D.~J. Costello, Jr., \emph{Error Control Coding: Fundamentals and
  Applications}, 2nd~ed.\hskip 1em plus 0.5em minus 0.4em\relax Englewood
  Cliffs, NJ: Prentice-Hall, Inc., 2004.

\bibitem{Luo&Min:13IJDSN}
Z.~X. Luo, P.~S. Min, and S.-J. Liu, ``Target localization in wireless sensor
  networks for industrial control with selected sensors,'' \emph{International
  Journal of Distributed Sensor Networks}, vol. 2013, 2013.

\bibitem{Vempaty_ICASSP13_loc}
A.~Vempaty, Y.~S. Han, and P.~K. Varshney., ``Target localization in wireless
  sensor networks using error correcting codes in the presence of
  {B}yzantines,'' in \emph{Proc. IEEE Int. Conf. Acoust., Speech, Signal
  Process. (ICASSP 2001)}, Vancouver, Canada, May 2013, pp. 5195--5199.

\bibitem{Vempaty_ISCIT13}
------, ``Byzantine tolerant target localization in wireless sensor networks
  over non-ideal channels,'' in \emph{{Proc. 13th Int. Symp. Commun. Inf.
  Technologies (ISCIT 2013)}}, Samui Island, Thailand, Sep. 2013.

\bibitem{Wang_jsac05}
T.-Y. Wang, Y.~S. Han, P.~K. Varshney, and P.-N. Chen, ``Distributed
  fault-tolerant classification in wireless sensor networks,'' \emph{IEEE J
  Sel. Areas Comm.}, vol.~23, no.~4, pp. 724 -- 734, April 2005.

\bibitem{Wang_twc06}
T.-Y. Wang, Y.~S. Han, B.~Chen, and P.~K. Varshney, ``A combined decision
  fusion and channel coding scheme for distributed fault-tolerant
  classification in wireless sensor networks,'' \emph{IEEE Trans. Wireless
  Commun.}, vol.~5, no.~7, pp. 1695--1705, 2006.

\bibitem{Berkhin02surveyof}
P.~Berkhin, ``Survey of clustering data mining techniques,'' Tech. Rep., 2002.

\bibitem{Aurenhammer:1991:VDS:116873.116880}
\BIBentryALTinterwordspacing
F.~Aurenhammer, ``Voronoi diagram- {A} survey of a fundamental geometric data
  structure,'' \emph{ACM Comput. Surv.}, vol.~23, no.~3, pp. 345--405, Sep.
  1991. [Online]. Available: \url{http://doi.acm.org/10.1145/116873.116880}
\BIBentrySTDinterwordspacing

\bibitem{Yao_TIT'07}
C.~Yao, P.-N. Chen, T.-Y. Wang, Y.~S. Han, and P.~K. Varshney, ``Performance
  analysis and code design for minimum {H}amming distance fusion in wireless
  sensor networks,'' \emph{IEEE Trans. Inf. Theory}, vol.~53, no.~5, pp. 1716
  --1734, May 2007.

\bibitem{feller}
W.~Feller, \emph{An Introduction to Probability Theory and Its
  Applications}.\hskip 1em plus 0.5em minus 0.4em\relax New York: Wiley, 1966.

\end{thebibliography}
\end{document}